\documentclass[12pt]{article} 
\usepackage[sectionbib]{natbib}
\usepackage{array,epsfig,fancyheadings,rotating}
\usepackage[]{hyperref}  
\usepackage{sectsty, secdot}
\sectionfont{\fontsize{12}{14pt plus.8pt minus .6pt}\selectfont}
\renewcommand{\theequation}{\thesection\arabic{equation}}
\subsectionfont{\fontsize{12}{14pt plus.8pt minus .6pt}\selectfont}

\usepackage{amsmath}
\usepackage{graphicx}
\usepackage{enumerate}
\usepackage{url}    
\usepackage{booktabs}      
\usepackage{amsfonts}       
\usepackage{nicefrac}      
\usepackage{microtype}    
\usepackage{xcolor}      
\usepackage{amsthm}
\usepackage{mathtools}
\usepackage{amssymb}
\usepackage{mathrsfs} 
\usepackage{caption}
\usepackage{parskip}
\usepackage{subcaption}
\usepackage{float}
\usepackage{algorithm}
\usepackage{algpseudocode}
\usepackage{geometry}
\usepackage{subcaption}
\usepackage{afterpage}

\newtheorem{theorem}{Theorem}
\newtheorem{lemma}{Lemma}
\newtheorem{corollary}{Corollary}[theorem]

\theoremstyle{definition}
\newtheorem{definition}{Definition}

\newtheorem{remark}{Remark}
\newtheorem{assumption}{Assumption}
\begin{document}
	
	
	\renewcommand{\baselinestretch}{2}
	
	\markright{ \hbox{\footnotesize\rm 
		}\hfill\\[-13pt]
		\hbox{\footnotesize\rm
		}\hfill }
	
	\markboth{\hfill{\footnotesize\rm Yuhan Tian and Abolfazl Safikhani} \hfill}
	{\hfill {\footnotesize\rm Online Change Point Detection in High-dimensional Vector Auto-regressive Models} \hfill}
	
	\renewcommand{\thefootnote}{}
	$\ $\par
	
	
	\fontsize{12}{14pt plus.8pt minus .6pt}\selectfont
\vspace{0.8pc}
\centerline{\large\bf Sequential Change Point Detection in}
\vspace{2pt} 
\centerline{\large\bf High-dimensional Vector Auto-regressive Models}
\vspace{.4cm} 
\centerline{Yuhan Tian\textsuperscript{1} (yuhan.tian@ufl.edu) and Abolfazl Safikhani\textsuperscript{2} (asafikha@gmu.edu)} 
\vspace{.4cm} 
\centerline{\it \textsuperscript{1}Department of Statistics, University of Florida, Gainesville, FL 32611} 
\centerline{\it \textsuperscript{2}Department of Statistics, George Mason University, Fairfax, VA 22030} 
\vspace{0.3cm}

\fontsize{9}{11.5pt plus.8pt minus.6pt}\selectfont

	
	\begin{quotation}
		\noindent {\it Abstract:}
		{Sequential/Online change point detection involves continuously monitoring time series data and triggering an alarm when shifts in the data distribution are detected. We propose an algorithm for real-time identification of alterations in the transition matrices of high-dimensional vector auto-regressive models. This algorithm initially estimates transition matrices and error term variances using regularization techniques applied to training data, then computes a specific test statistic to detect changes in transition matrices as new data batches arrive. We establish the asymptotic normality of the test statistic under the scenario of no change points, subject to mild conditions. An alarm is raised when the calculated test statistic exceeds a predefined quantile of the standard normal distribution. We demonstrate that as the size of the change (jump size) increases, the test's power approaches one. Empirical validation of the algorithm's effectiveness is conducted across various simulation scenarios. Finally, we discuss two applications of the proposed methodology: analyzing shocks within S\&P 500 data and detecting the timing of seizures in EEG data.}
		
		\vspace{9pt}
		\noindent {\it Key words and phrases:}
		auto correlation; break point; sequential data; structural break; temporal dependence.
		\par
	\end{quotation}\par

	\def\thefigure{\arabic{figure}}
	\def\thetable{\arabic{table}}
	
	\renewcommand{\theequation}{\thesection.\arabic{equation}}

	\fontsize{12}{14pt plus.8pt minus .6pt}\selectfont

\newpage

\section{INTRODUCTION}\label{introduction}

\lhead[\footnotesize\thepage\fancyplain{}\leftmark]{}\rhead[]{\fancyplain{}\rightmark\footnotesize\thepage}

Abrupt changes in daily life are often perceived as anomalies, typically requiring careful study to prevent future repercussions. In a data set, such abrupt changes are usually triggered by shifts in the data-generating process. Detecting these changes precisely and quickly is essential for understanding their origins and mitigating potential harm. Consequently, change point detection (CPD) has become a critical research area in both data science and statistics, with real-world applications spanning power systems, quality control, and advertising \citep{routtenberg2017pmu, page1954continuous, zhang2017online}. Most CPD methods, classified here as offline CPD, require access to the full dataset and aim to pinpoint change point locations accurately. However, with advancements in cloud storage and computing, streaming data has become ubiquitous, necessitating a different CPD approach for monitoring incomplete and dynamic data in real time. Online (or Sequential) CPD addresses this need by triggering alarms as changes are detected in data streams. A robust online CPD method should therefore achieve low false alarm rates, minimal detection delays, and efficient computational processing. In this study, we introduce an online change point detection algorithm that meets all these requirements.

A wide range of online CPD methods has been documented in the literature, primarily focusing on techniques to detect changes in distribution parameters of univariate data, such as the mean, variance, or overall distribution. These methods include early applications of Shewhart charts, cumulative sum (CUSUM) charts, and exponentially weighted moving average (EWMA) charts for quality control \citep{shewhart1930economic, page1954continuous, roberts2000control}, as well as more recent advances based on likelihood ratio tests \citep{hawkins2005statistical}. For more comprehensive coverage of control charts for univariate and multivariate time series, see \cite{montgomery2019introduction, qiu2013introduction}. Recent advances in computational power and data storage have enabled broader interest in multivariate time series, with applications across finance, weather forecasting, healthcare, and industrial operations. Developing online CPD algorithms for multivariate (or high-dimensional) data introduces two main challenges: adapting univariate test techniques for multivariate data and accounting for inter-component correlations, both contemporaneous and cross-correlated. These challenges complicate the theoretical guarantees for false alarm control and detection delay and require careful attention to computational efficiency in an online setting. Two common solutions include applying univariate CPD methods independently to each series or transforming the multivariate time series into a single series for univariate CPD analysis \citep{jandhyala2013inference}. The former approach may weaken detection power by overlooking common change points, while the latter may be ineffective if change points in some series are masked by noise in the combined data. To address this, several methods aggregate component-wise test statistics rather than observations \citep{chen2022high, gosmann2022sequential, bardwell2019most, xu2021optimum}. For instance, the algorithm in \cite{bardwell2019most} examines each series individually to generate a profile-likelihood-like statistic for change points and post-processes the results to identify common change points. However, this method lacks theoretical guarantees for false positives and detection delays. Another approach by \cite{chen2022high} uses likelihood ratio tests across scales and coordinates, but it assumes independent Gaussian observations, limiting its applicability to real-world data with temporal dependencies and cross-correlations. Additionally, \cite{gosmann2022sequential} integrates component-wise schemes using a maximum statistic, which converges to a Gumbel distribution as dimension and sample size increase, though this method only detects mean shifts and not changes in variance or covariance. Finally, \cite{xu2021optimum} introduces an online CPD approach with sampling control, selecting only a few observable series at each time point and employing a sequential probability ratio test. This approach reduces dimensionality by focusing on selected series, though it assumes independent and identically distributed samples. These methods, while capable of handling high-dimensional settings, struggle with contemporary and cross correlations inherent in multivariate time series due to their reliance on independent test statistics per series.

For online CPD in data with dependencies and cross-correlations, treating the multivariate series as an integrated entity is more viable. In Bayesian Online CPD \citep{adams2007bayesian}, change point inference is based on the posterior distribution of the current run length, updated sequentially with new data. This method is flexible through its choice of predictive distribution, but it is not readily adaptable to high-dimensional data, where the likelihood becomes computationally infeasible as dimensions increase. Additionally, Bayesian Online CPD’s time complexity scales linearly with the number of observations, making it unsuitable for long time series (in contrast, our algorithm's complexity is independent of the number of observations, depending only on window size). Several recent methods employ graph-based techniques for online CPD in high-dimensional data. For instance, k-nearest neighbor (k-NN) algorithms by \cite{chen2019sequential, chu2022sequential} perform two-sample tests on k-NN sequentially as data arrives. These algorithms apply to high-dimensional series with contemporaneous correlations, given a suitable similarity measure. However, temporal dependencies can undermine k-NN methods, as the local neighborhood’s definition becomes unreliable over time. Specifically, the choice of neighbors may change as patterns shift, making it difficult to adapt effectively to temporal dependencies. Projection-based control charts offer a practical approach for handling high-dimensionality and correlations in process monitoring. A notable example is \cite{zhang2020spatial}, which uses random projections to reduce dimensionality, creating subprocesses that are monitored by local nonparametric control charts and then fused for decision-making. PCA-based control charts provide another option for dimension reduction, addressing various high-dimensional data types \citep{de2015overview}. For example, dynamic PCA-based charts \citep{ku1995disturbance} manage autocorrelation by including lagged data, while recursive PCA charts \citep{choi2006adaptive} handle nonstationarity by updating parameters with a forgetting factor, and moving window PCA charts \citep{he2008variable} maintain a recent data window. However, projection-based methods can lack interpretability since identified changes may involve multiple variables, complicating the source identification. For a detailed overview of CPD methods, see \cite{aminikhanghahi2017survey}. Despite the extensive work on online CPD, few methods effectively manage high-dimensional settings with cross correlations and even fewer offer theoretical guarantees, highlighting the need for our proposed algorithm tailored to these challenges.

To address high-dimensional data with temporal and cross correlations, our algorithm is based on the vector auto-regressive (VAR) model, represented in equation (\ref{eqn:a1}). The key parameters of a VAR model are its transition matrices, which capture temporal and cross dependencies among observations. This linear structure provides computational and analytical efficiency, making VAR a staple in multivariate time series analysis, with applications spanning economics \cite{rosser1995vector}, neuroscience \cite{goebel2003investigating}, and quality control \cite{pan2012and}. Changes within a VAR time series manifest as alterations in its transition matrices, as described in \cite{wang2019localizing}, enabling our algorithm to detect shifts in higher-order structures, including temporal and cross correlations. This capability differentiates our approach, as most online CPD algorithms focus on changes in mean \citep{gosmann2022sequential, dette2020likelihood, hawkins2005statistical, aminikhanghahi2018real}, variance \citep{dette2020likelihood, hawkins2005statistical, aminikhanghahi2018real}, or contemporary correlations \citep{dette2020likelihood, zhang2023spectral, cabrieto2017detecting}. With the rise of high-dimensional data, VAR models have become increasingly important. Despite their wide applications, a clear gap remains: to our knowledge, no online or sequential CPD algorithms are explicitly designed for high-dimensional VAR models. Our algorithm aims to bridge this gap by providing an online CPD approach for detecting abrupt changes in transition matrices in high-dimensional VAR models.

To explain our algorithm, we introduce two essential quantities: $n$ and $\omega$. Here, $n$ represents the number of observations in our training data set (historical data set), carefully selected to exclude any change points—a common practice in online CPD research \citep{qiu2022transparent, gosmann2022sequential, chen2022high}. The parameter $\omega$ denotes the permissible detection delay, as data is observed incrementally in an online setting. To assess whether time $t$ is a change point, a few subsequent data points are needed, referred to as the “pre-specified detection delay” ($\omega$). Following \cite{aminikhanghahi2017survey}, our algorithm is therefore a $\omega$-real-time algorithm. Our algorithm has two primary steps. First, we estimate transition matrices and error variances using the training data set, denoted as $X_1, X_2, \ldots, X_n$. Given the high-dimensional nature of the model, we use an $\ell_1$-penalized least squares estimator, which needs to be computed only once for the entire monitoring process. Second, we compute a test statistic on sequential data batches of size $\omega$, specifically $X_{t+1}, X_{t+2}, \ldots, X_{t+\omega}$, observed at times $t, t+1, \ldots$. Our method triggers an alarm if the test statistic exceeds a predefined threshold (see Section~\ref{statistics} for details). Additionally, we include a refinement step to locate change points and reduce false alarms (see Section~\ref{refine}). In Theorem~\ref{th:1}, we establish the asymptotic normality of the test statistic under conditions without change points, and in Theorem~\ref{th:2}, we examine the relationship between the power of the test and the jump size (the difference in model parameters before and after the change point). This analysis requires examining the fourth-order properties of the VAR time series, as parameter consistency in high-dimensional settings generally involves only first- and second-order properties \citep{basu2015regularized, wong2020lasso}. The proof is detailed in Section~\ref{Appendix C} of the Supplementary Materials. The normality of the test statistic allows for selecting a threshold to control the average run length (false alarm rate) without resorting to costly Monte Carlo simulations, a step often required by other online CPD algorithms \citep{chen2022high, qiu2022transparent}. This is particularly beneficial in high-dimensional online scenarios. Our algorithm demonstrates robust numerical performance, achieving well-controlled average run length and short detection delay, as detailed in Section~\ref{section:num_main}. It also surpasses several alternative methods in computational efficiency, an advantage supported by results in Sections~\ref{section:num_main} and \ref{realdata}. This speed is critical for online CPD applications. Additionally, the algorithm is resource-efficient, requiring only moderate memory and data storage for parameter estimates.

Our algorithm's main contributions are as follows. It is an online change point detection algorithm specifically designed for high-dimensional VAR models, addressing a notable gap in the field. It detects changes in higher-order structures, such as temporal and cross correlations, while handling high dimensionality and dependencies. The algorithm has theoretical guarantees, with asymptotic normality allowing control over the average run length. We also analyze the link between algorithm power and jump size. It optimizes resource usage, reducing the need for Monte Carlo simulations during threshold selection and enhancing the efficiency of the CPD process.

The paper is organized as follows: Section~\ref{formulation} describes the model setup and change point detection problem. The test statistic and detection algorithm are outlined in Section~\ref{algorithm}, and theoretical results are presented in Section~\ref{theorem}. Section~\ref{refine} introduces a refinement approach for precise change point estimation. Section~\ref{multicp} extends the algorithm to multiple change points. Performance results on synthetic data are presented in Section~\ref{section:num_main}, while real data applications are discussed in Section~\ref{realdata}. Finally, Section~\ref{future} provides concluding remarks and future research directions.

\emph{Notations:} In this paper, when referring to a vector \(v \in \mathbb{R}^{p}\), we denote its \(j\)-th feature as \(v_j\). The \(\ell_q\) norms are represented by \(\|v\|_{q} = \left(\sum_{j=1}^{p}\left|v_{j}\right|^{q}\right)^{1 / q}\), where \(q > 0\). We use \(\|v\|_{0}\) to denote \(\left|\operatorname{supp}(v)\right| = \sum_{i=1}^{p} \mathbf{1}\left[v_{j} \neq 0\right]\), and \(\|v\|_{\infty}\) to represent \(\max_{j}\left|v_{j}\right|\). For an indexed vector, \(v_i \in \mathbb{R}^{p}\) for \(i = 1, \ldots, n\), its \(j\)-th feature is denoted as \(v_{i,j}\). In the case of a matrix \(A\), \(\rho(A)\), \(\left\|A\right\|\), and \(\left\|A\right\|_{F}\) denote its spectral radius \(\left|\Lambda_{\max }(A)\right|\), operator norm \(\sqrt{\Lambda_{\max }\left(A^{\prime} A\right)}\), and Frobenius norm \(\sqrt{\operatorname{tr}\left(A^{\prime} A\right)}\), respectively. Additionally, \(\left\|A\right\|_{\max}\), \(\left\|A\right\|_{\infty}\), and \(\left\|A\right\|_{1}\) denote the coordinate-wise maximum (in absolute value), maximum absolute row sum, and maximum absolute column sum of matrix \(A\), respectively. For matrix \(A\), its maximum and minimum eigenvalues are represented as \(\Lambda_{\min }(A)\) and \(\Lambda_{\max }(A)\), and the element in the \(i\)-th row and \(j\)-th column is denoted as \(a_{i,j}\). The \(i\)-th unit vector in \(\mathbb{R}^{p}\) is indicated as \(e_{i}\). Throughout this paper, the notation \(A \succsim B\) signifies the existence of an absolute constant \(c\), independent of model parameters, such that \(A \geq c B\). The notation \(A \asymp B\) is used to denote \(A \succsim B\) and \(B \succsim A\). The notation \(\Gamma_{X}\left(\ell\right)\) is used to represent Cov\(\left(X_t,X_{t+l}\right)\). \(A^\prime\) and \(A^*\) stand for the ordinary transpose and the Hermitian transpose of matrix \(A\) respectively. Convergence in probability and in distribution is indicated by \(\xrightarrow[]{p}\) and \(\xrightarrow[]{D}\) respectively. The variance of a random variable or vector \(x\) is denoted as \(\operatorname{VAR}\left(x\right)\).

\vspace{-0.5cm}

\section{MODEL FORMULATION}\label{formulation}
We assume the $p$-dimensional data are generated from a finite order VAR(h) process with the transition matrices $A_1,\ldots, A_h$ before the occurrence of the change. After the change point, i.e., from time ${t^*+1}$, transition matrices are changed to $A^*_1,\ldots, A^*_h$ with $\left( A^*_1,\ldots, A^*_h \right) \neq \left( A_1,\ldots, A_h \right)$. The lag $h$ may vary before and after the change point. In such situations, we refer to $h$ as the maximum of the two lags, and we augment the process with the smaller lag by including a few zero-matrix transition matrices. Consequently, to simplify matters, we assume that the lag $h$ remains constant both before and after the change point without loss of generality. Formally, data points before the change point, denoted as $\left\{\ldots, X_{t^*-1}, X_{t^*}\right\}$ (including the training set $\left\{X_{-h+1}, \ldots, X_n\right\}$), and data points subsequent to the change point, denoted as $\left\{X_{t^*+1}, X_{t^*+2}, \ldots \right\}$, are generated according to the following equations:
\begin{equation}\label{eqn:a1}
    X_t = \sum_{l=1}^h A_l X_{t-l} + \varepsilon_t, \,\,\, \text{for}\,\, t \leq t^*; \, X_t = \sum_{l=1}^h A^*_l X_{t-l} + \varepsilon_t, \,\,\, \text{for}\,\, t > t^*.
\end{equation}

Here, the error vectors $\varepsilon_t$ are temporally independent, possessing a mean of zero and a covariance matrix denoted as $\Sigma = \sigma^2 I_p$. We refer to Remark~\ref{remark1} for a more flexible structure for the covariance of the error term. The VAR model naturally entails a high-dimensional parameter space, as the number of parameters scales quadratically with respect to data dimension (i.e. the number of parameters is of order \(p^2\) where $p$ is the data dimension). Even when \(p\) is relatively modest, this scaling leads to a substantial parameter count, exposing the model to the challenges of high-dimensional settings.


\vspace{-0.5cm}

\section{DETECTION ALGORITHM}\label{algorithm}

In this section, we provide details of the proposed detection algorithm. The algorithm consists of two main steps. We assume we have access to $n + h$ training data points in which there are no changes in transition matrices. In the first step, these data points are used to estimate the baseline transition matrices and variance of error terms. In the second step, new batches of observations of size $\omega$ are observed, and test statistics are computed using these batches and model parameter estimations from the first step. Large values of the test statistic indicate a potential change point.

\vspace{-0.5cm}

\subsection{Step I: Estimation of Transition Matrices and Error Variance}\label{estimation}
In this step, we aim to estimate the transition matrices and the variance of error terms using the provided training data. To achieve this, we construct a regression problem based on the training data denoted as $\mathbb{X}_{hist} = \left\{X_{-h+1}, \ldots, X_n\right\}$. This regression problem takes the following form:
\[
\underbrace{\left[\begin{array}{c}
X_{n}^{\prime} \\
\vdots \\
X_{1}^{\prime}
\end{array}\right]}_{\mathcal{Y}_n}=\underbrace{\left[\begin{array}{ccc}
X_{n-1}^{\prime} & \cdots & X_{n-h}^{\prime} \\
\vdots & \ddots & \vdots \\
X_{0}^{\prime} & \cdots & X_{-h+1}^{\prime}
\end{array}\right]}_{\mathcal{X}_n} \underbrace{\left[\begin{array}{c}
A_{1}^{\prime} \\
\vdots \\
A_{h}^{\prime}
\end{array}\right]}_{B^{*}}+\underbrace{\left[\begin{array}{c}
\varepsilon_{n}^{\prime} \\
\vdots \\
\varepsilon_{1}^{\prime}
\end{array}\right]}_{E_n}.
\]
This problem can be expressed in vector form as: $\operatorname{vec}(\mathcal{Y}_n) =\operatorname{vec}\left(\mathcal{X}_n  B^*\right)+\operatorname{vec}(E_n)=(I \otimes \mathcal{X}_n) \operatorname{vec}\left(B^*\right)+\operatorname{vec}(E_n)$. Alternatively, it can be represented as: $$
\underbrace{Y_n}_{n p \times 1} =\underbrace{Z_n}_{n p \times {hp^2}} \underbrace{\beta^{*}}_{{hp^2} \times 1}+\underbrace{\operatorname{vec}(E_n)}_{n p \times 1}.
$$

We employ an $\ell_1$-penalized least squares approach to estimate the transition matrices $A_1, A_2, \ldots, A_h$, which is equivalent to estimating $\beta^{*}$. Simultaneously, we estimate $\sigma^2$ and $\operatorname{Var}(\varepsilon^2_{1,1})$ using the method of moments:
\begin{equation}\label{eqn:a3}
\hat{\beta}_n = \underset{\beta \in \mathbb{R}^{hp^2}}{\operatorname{argmin}} \: \left(\frac{1}{n}\|Y_n-Z_n \beta\|_2^{2}+\lambda_{n}\|\beta\|_{1}\right),  
\end{equation}
\begin{equation}
\hat{\sigma}^2_n = \frac{1}{pn}\sum_{i=1}^{n}\left\|\left(\sum_{l=1}^h\hat{A}_l X_{i-l}\right) - X_i\right\|_2^2 ,
\end{equation}
\begin{equation}
\text{and }\hat{V}_n = \left|\frac{1}{pn}\sum_{i=1}^n\left\|\left(\sum_{l=1}^h\hat{A}_l X_{i-l}\right)-X_i\right\|_4^4 - \hat{\sigma}^4_n\right|.
\end{equation}

The estimator for the transition matrices, employing $\ell_1$-penalized least squares, exhibits several valuable properties, including consistency in high-dimensional settings \citep{basu2015regularized}. The parameter $\lambda_n$ acts as a tuning parameter, controlling sparsity in the estimation. The choice of the tuning parameter $\lambda_n$ is determined through cross-validation, and additional information can be found in the R package ``sparsevar," as introduced in \cite{sparsevar}. Finally, if the lag $h$ for the VAR process is unknown, We recommend estimating it by comparing the Bayes information criterion (BIC), defined as $\operatorname{BIC}(h)=\ln |\hat{\Sigma}(h)|+\frac{\ln n}{n} h p^{2},\text{ where }\hat{\Sigma}(h)=n^{-1} \sum_{t=1}^{n} \hat{\varepsilon}_{t,h} \hat{\varepsilon}_{t,h}^{\prime}$. Here, $\hat{\varepsilon}_{t,h}$ represents the residual at time $t$ when a VAR(h) model is employed. To determine the appropriate lag, one should calculate BIC($h$) for a grid of VAR models with various potential lags, using the historical data. The lag with the lowest BIC($h$) should be selected as the estimated lag.

\vspace{-0.5cm}

\subsection{Step II: Test Statistic}\label{statistics}
Given the parameter estimates $\hat{\beta}_n$, $\hat{\sigma}^2_n$, $\hat{V}_n$, and the new observations $\mathbb{X}_{obs} = {X_{t+1}, \ldots, X_{t+\omega}}$ with $t > n - \omega$, we define the test statistic as follows: 
\begin{equation}\label{eqn:a4}
\hat{T}_t^{(n,\omega)} =\sqrt{\frac{p\omega}{\hat{V}_n}}\left(\frac{\hat{R}_t^{(n,\omega)}}{p}-\hat{\sigma}^2_n\right),
\end{equation}
\begin{equation}
\text{ where }\hat{R}_t^{(n,\omega)} = \frac{1}{\omega}\sum_{i = t + 1}^{t + \omega} \left\|\left(\sum_{l=1}^h\hat{A}_l X_{i-l}\right) - X_i\right\|^2_2.
\end{equation}
Finally, we compute the test statistic $\hat{T}_t^{(n,\omega)}$ for $t = n - \omega + 1, n - \omega +2, \ldots$. An alarm will be raised at time $\hat{t}$ if $\left|\hat{T}_{\hat{t}}^{(n,\omega)}\right| > \Phi(1-\alpha/2)$, where $\Phi(\cdot)$ represents the standard normal quantile function.

\begin{algorithm}[!ht]
\caption{VAR\_cpDetect\_Online}\label{alg1}
\begin{algorithmic}
\Require $data \in \mathbf{R}^{p \times T}$, $n$, $\omega$, $\alpha$, $h$
\State $t \gets 0$; $\mathbb{X}_{hist} \gets data[,t + 1:t + n + h]$; $\hat{\beta}_n \gets \underset{\beta \in \mathbb{R}^{hp^2}}{\operatorname{argmin}} \: \left(\frac{1}{n}\left\|Y_n-Z_n \beta\right\|^{2}_2+\lambda_{n}\left\|\beta\right\|_{1}\right)$
\State $\hat{\sigma}^2_n \gets \frac{1}{pn}\sum_{i=1}^{n}\left\|\left(\sum_{l=1}^h\hat{A}_l X_{i-l}\right) - X_i\right\|^2_2$; 
\State $\hat{V}_n \gets \left|\frac{1}{pn}\sum_{i=1}^n\left\|\left(\sum_{l=1}^h\hat{A}_l X_{i-l}\right)-X_i\right\|_4^4 - \hat{\sigma}^4_n\right|$; $t \gets n+h-\omega+1$
\While{$t \leq T-\omega$}
\State $\mathbb{X}_{obs} \gets data[,t - h + 1: t + \omega]$; $\hat{T}_t^{(n,\omega)} \gets \sqrt{\frac{p\omega}{\hat{V}_n}}\left(\frac{\hat{R}_t^{(n,\omega)}}{p}-\hat{\sigma}^2_n\right)$ 
\If{$\left|\hat{T}_t^{(n,\omega)}\right| > \Phi(1-\alpha/2)$}
    \State raise alarm; $t \gets t + 1$ 
\Else
    \State $t \gets t + 1$
\EndIf
\EndWhile
\end{algorithmic}
\end{algorithm}

The underlying concept behind the developed test statistic lies in its behavior under different scenarios. When no change points are present, our test statistic becomes a normalized sum of independent and identically distributed random variables, assuming our estimation of transition matrices is consistent. In such cases, the distribution of the test statistic closely approximates a standard normal distribution, as demonstrated in Theorem~\ref{th:1}. Conversely, if a change point exists before time $t$, our test statistic exhibits a shift, as described in Theorem~\ref{th:2}. Consequently, an alarm will be raised when $\left|\hat{T}_t^{(n,\omega)}\right| > \Phi(1-\alpha/2)$. The selection of \( \omega \) is discussed in detail in Section~\ref{sim:w} of the supplementary material. Identifying which components experience shifts after a raised alarm is crucial in high-dimensional settings. We propose using an online debiasing technique \citep{deshpande2023online} to construct confidence intervals for the differences in transition matrices before and after the change to infer the changing components; further details are provided in Section~\ref{postchangeanalysis} of the supplementary material.
\begin{remark}\label{remark1}
    It is possible to extend the proposed algorithm to accommodate scenarios with variance heterogeneity. In such cases, the modified test statistic is defined as: 
    \[
    \hat{T}_t^{(n,\omega)} = \frac{\sqrt{\omega}\left(\hat{R}^{(n,\omega)}_t-\sum^p_{j=1}\hat{\sigma}^2_{n,j}\right)}{\sqrt{\sum^p_{j=1}\hat{V}_{n,j}}},  
    \]
    where $\hat{\sigma}^2_{n,j}$ and $\hat{V}_{n,j}$ are estimated separately using the method of moments for each component $j$. In order to keep the exposition of proposed methodology clear, we focus on fixed/homogeneous variance case in the remainder of the paper while the satisfactory performance under heterogeneous case is empirically illustrated in Section~\ref{sim:hetervar} in supplementary material. Note that extending the algorithm to scenarios with non-diagonal covariance matrices is discussed in Section~\ref{future}. 
\end{remark}

\vspace{-0.5cm}

\section{THEORETICAL PROPERTIES}\label{theorem}
In this section, we present two theorems concerning the asymptotic behavior of the test statistic in two distinct scenarios: one when there are no change points, and the other when a change point exists. To derive these theorems, it is necessary to make the following assumptions.

\begin{assumption}\label{assumption1}
    \emph{The transition matrices exhibit sparsity, meaning that the vector $\beta^*$ possesses a sparsity level denoted as $s$, represented as $\|\beta^*\|_0 = s$.}
\end{assumption}
\begin{assumption}\label{assumption2}
    \emph{The error terms, denoted as $\varepsilon_t$, are independent sub-Gaussian random vectors with a mean of zero and a variance of $\sigma^2I_p$. Moreover, their sub-Gaussian norm is bounded by a constant $K$.}
\end{assumption}
\begin{assumption}\label{assumption3}
    \emph{The VAR process is stable and stationary.}
\end{assumption}
\begin{assumption}\label{assumption4}
    \emph{The spectral density function, denoted as $f_{X}(\theta):=\frac{1}{2 \pi} \sum_{\ell=-\infty}^{\infty} \Gamma_{X}(\ell) e^{-i \ell \theta}$, exists for $\theta$ within the interval $[-\pi, \pi]$. Additionally, its maximum and minimum eigenvalues are bounded on this interval, that is,
\[
\mathcal{M}\left(f_{X}\right):=\operatorname{sup}_{\theta \in[-\pi, \pi]} \Lambda_{\max }\left(f_{X}(\theta)\right)<\infty \text{ and } \mathfrak{m}\left(f_{X}\right):=\operatorname{inf}_{\theta \in[-\pi, \pi]} \Lambda_{\min }\left(f_{X}(\theta)\right)>0.
\]}
\end{assumption}

Assumption~\ref{assumption1} is common in high-dimensional models and plays a important role in addressing dimensionality issues. On the other hand, Assumption~\ref{assumption2} is employed to manage the tail behavior of the data distribution. It's worth noting that the sub-Gaussian assumption can be relaxed to accommodate heavier-tailed distributions, such as the sub-Weibull distribution, albeit at the expense of a larger sample size requirement, as discussed in \cite{wong2020lasso}. Assumption~\ref{assumption3} is common in the time series literature, as seen in references like \cite{lutkepohl2005new}, and it ensures the existence of a unique stationary solution for the auto-regressive equations \eqref{eqn:a1}. Finally, Assumption~\ref{assumption4} is essential for verifying the restricted eigenvalue and deviation bound conditions, as outlined in \cite{loh2012high} and \cite{basu2015regularized}. These two properties are crucial prerequisites for establishing the consistency of the $\ell_1$-regularized estimates in \eqref{eqn:a3}.

\begin{theorem}\label{th:1}
\emph{Suppose that there are no change points in the data generation process, and Assumptions~\ref{assumption1}-\ref{assumption4} are satisfied. Then, with $\omega = o(n)$, $\omega\succsim s(\log h + 2 \log p)$,
 $\frac{s(\log h + 2 \log p)}{\sqrt{p}}=o(\sqrt{n})$, $n^{1/2-a}\succsim\frac{\sqrt{s}}{p^{1/2-a}}$ for some $a \in (0,1/2)$ and $n^{1/4-b}\succsim\frac{\sqrt{s}}{p^{3/4-b}}$ for some $b \in (0,1/4)$, we have
\[
\hat{T}_t^{(n,\omega)} \xrightarrow[]{D} \mathcal{N}(0,1)\,\,\, \text{as} \,\, n  \xrightarrow{} \infty,
\]
where $\mathcal{N}(0,1)$ represents the standard normal distribution.}
\end{theorem}

This theorem forms the foundation of the proposed online detection algorithm by analyzing the marginal distribution of the test statistic in scenarios without change points. The asymptotic normality provides an objective basis for selecting the alarm threshold by utilizing the quantile function of the standard normal distribution. Consequently, the algorithm's average run length (ARL) can be controlled through the selection of the threshold $\alpha$. Note that the dependence arising from overlapping windows may impact the ARL or false alarm rate of the algorithm, and thus, proper choice of the threshold \( \alpha \) should be provided. As summarized in Section~\ref{sim:runlength} of the supplementary material, the proposed method of selecting \( \alpha \) (which does not account for potential overlapping dependence) has proven sufficient to control the ARL and maintain the false alarm rate. Typically, when the historical data set is sufficiently large for accurate parameter estimation, the average run length will be lower bounded by $1/\alpha$. The presence of a change in the model parameters will be indicated by significant deviations of the test statistic beyond a chosen quantile of the standard normal distribution. It is important to mention that the sample size requirements outlined in Theorem~\ref{th:1} are relatively lenient, allowing for the consideration of high-dimensional scenarios, provided that the transition matrices are sparse. For example, when dealing with a fixed lag $h$, it is possible to select values such as $p=n^c$ and $\omega = \left( \log n \log p \right)^{1+\epsilon} $, where $c$ and $\epsilon$ are positive constants. This choice remains valid as long as the sparsity level $s$ satisfies $s = o \left( \left( \log n \right)^{1+\epsilon} \left( \log p \right)^{\epsilon} \right)$.

\begin{theorem}\label{th:2}
\emph{Assume the existence of a change point at time $t^*$, and assume that Assumptions~\ref{assumption1}-\ref{assumption4} hold for the data both before and after this change point. Under the same conditions in Theorem~\ref{th:1} with additional conditions that $s(\log h + 2 \log p) = o(\omega)$, $\sqrt{\frac{s(\log h + 2 \log p)}{\omega}} = o\left(\left\|\beta^*-\beta_{new}\right\|_2\right)$ and $\omega^\eta p^\eta  \sqrt{\frac{s^3(\log h + 2 \log p)}{n}} = o\left(\left\|\beta^*-\beta_{new}\right\|_2\right)$ for some $\eta \in (0,1/4)$, we have
$$
\begin{aligned}
    \text{P}(Z_{t^*+h}^{(n,\omega)} +  \frac{c_{l}}{\sqrt{\hat{V}_n}}\sqrt{\frac{\omega}{p}}\left\|\beta^*-\beta_{new}\right\|^2_2 + L_{t^*+h}^{(n,\omega)} \leq \hat{T}_{t^*+h}^{(n,\omega)} \\ 
    \leq Z_{t^*+h}^{(n,\omega)} +  \frac{c_{u}}{\sqrt{\hat{V}_n}}\sqrt{\frac{\omega}{p}}\left\|\beta^*-\beta_{new}\right\|^2_2 + \left(L_{t^*+h}^{(n,\omega)}\right)^\prime) \geq 1 - \epsilon_{n,p,\omega},
\end{aligned}
$$
where 
$$
\begin{aligned}
&Z_{t^*+h}^{(n,\omega)} \xrightarrow[]{D} \mathcal{N}(0,1)\,\,\, \text{as} \,\, n  \xrightarrow{} \infty, \hat{V}_n \xrightarrow[]{p}\operatorname{Var}(\varepsilon^2_{1,1}) \text{ as } n \xrightarrow[]{}\infty,\\
&L_{t^*+h}^{(n,\omega)} = o_p\left(\sqrt{\frac{\omega}{p}}\left\|\beta^*-\beta_{new}\right\|^2_2\right), \left(L_{t^*+h}^{(n,\omega)}\right)^\prime = o_p\left(\sqrt{\frac{\omega}{p}}\left\|\beta^*-\beta_{new}\right\|^2_2\right) \text{ and }\\
&\epsilon_{n,p,\omega} = c_1 \exp(-c_2\omega) + c_{3} \exp \left[-c_{4}(\log h + 2 \log p)\right] \\
&+ c_{5}\exp \left(-c_6 n\right) + 2\exp (-c_{7}p^\eta \omega^\eta + \log (\omega p h)),
\end{aligned}
$$
for some positive constants $c_1,\ldots,c_7, c_{l},c_{u}$ where $\beta^*$ and $\beta_{new}$ denote the vectorized transition matrices before and after the change point, respectively.}
\end{theorem}
\begin{remark}
    The conditions introduced, which relate to the jump size denoted as $\left\|\beta^*-\beta_{new}\right\|_2$, are fundamental for evaluating the power of our test. Similar conditions are commonly found in the literature on change point detection and have been employed in various studies, such as assumption A3 in \cite{safikhani2022joint} and H2 in \cite{chan2014group}. It is important to highlight that the flexibility of selecting a small value for $\eta$ ensures that these conditions remain valid even in high-dimensional scenarios, as the term $p^\eta$ can be controlled.
\end{remark}
This theorem sheds light on the behavior of the test statistic in the presence of a change point. With a large sample size, our test statistic will have a lower bound that corresponds to a right-shifted standard normal distribution with high probability. The extent of this shift is influenced by the jump size. Furthermore, if we define $\tilde{t}$ as the last observation when the alarm is correctly raised $\left(\text{i.e., }\min\left\{t>t^*-\omega:\,\left|\hat{T}_t^{(n,\omega)}\right|>\Phi(1-\alpha/2)\right\} + \omega\right)$, we can establish the following corollary.

\begin{corollary}\label{coro1}
    \emph{Under the same conditions as outlined in Theorem~\ref{th:2}, for any $k>0$ and a sufficiently large $n$, there exists $\epsilon^{(n)} = o(1)$. If $\frac{c_{l}}{2\sqrt{\operatorname{Var}(\varepsilon^2_{1,1})}}\sqrt{\frac{\omega}{p}}\left\|\beta^*-\beta_{new}\right\|_2^2 > \Phi(1-\alpha/2) + k$, then the following inequalities hold:
    $$
\begin{aligned}
 &P\left(\tilde{t}-t^* \leq \omega + h\right) \geq
P\left(\left|\hat{T}_{t^*+h}^{(n,\omega)}\right|>\Phi(1-\alpha/2)\right) \geq 1 - \epsilon_{n,p,\omega} - \epsilon^{(n)} - \exp(-k^2/2).
\end{aligned}
$$
}
\end{corollary}

 As demonstrated in this corollary, when dealing with a substantial sample size and a considerable jump size, the detection delay $\tilde{t}-t^*$ is likely to be upper bounded by $\omega + h$. Furthermore, the power of our test, denoted as $P\left(\left|\hat{T}_{t^*+h}^{n,\omega}\right|>\Phi(1-\alpha/2)\right)$, can approach one as the jump size increases and the sample size grows. The proof for Corollary~\ref{coro1} relies on Theorem~\ref{th:2} and the concentration inequality for the standard normal distribution.

\vspace{-0.5cm}

 \section{CHANGE POINT LOCALIZATION}\label{refine}
Currently, our algorithm can only trigger an alarm when it detects change points within the observations contained in a window of size $\omega$. However, determining the precise location of the change point remains unresolved. This situation is commonly encountered in the literature, see e.g.  \cite{chen2022high,mei2010efficient, xie2013sequential, chan2017optimal}. Due to the limited number of observations available after the change point, accurately pinpointing its exact location proves to be challenging. Consequently, we propose a potential solution to refine the estimated location produced by our algorithm. The idea involves re-executing our algorithm using a smaller pre-specified detection delay of $\omega^\prime$ within the current data window of size $\omega$ after an alarm is triggered. More specifically, upon the alarm being triggered by our algorithm at time $\hat{t}$, we treat the observations from $\hat{t}+1$ to $\hat{t}+\omega$ as the new data set that needs to be monitored. We will then apply our algorithm to this data set, employing a reduced pre-specified detection delay of $\omega^\prime$, and record the resulting estimated location of the change point as the refined estimate ($\hat{\hat{t}}$). In this scenario, our theorems remain valid if $\omega^\prime$ satisfies the conditions for $\omega$ (for example, one can set $\omega^\prime = \omega/5$). Figure~\ref{fig:localillu} provides a visual demonstration of this process.
\begin{figure}[!ht]
\begin{center}
\includegraphics[width=0.8\textwidth]{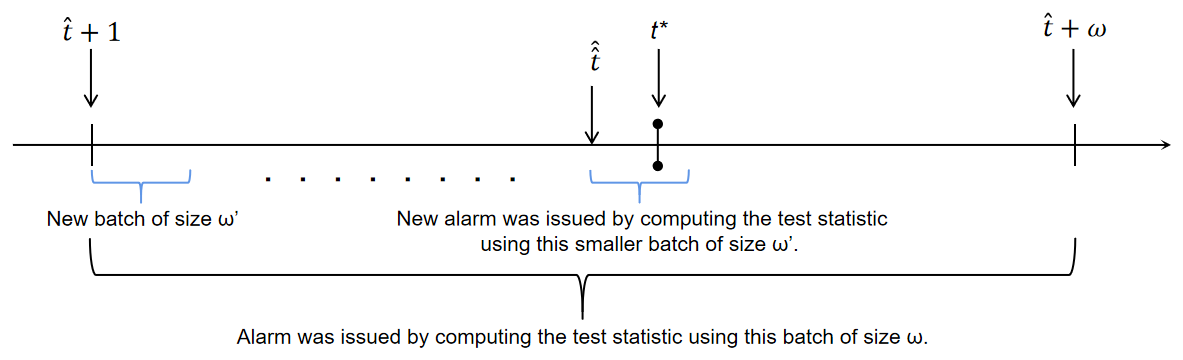}
\end{center}
\caption{Illustration of Refinement}
\label{fig:localillu}
\end{figure}

\vspace{-0.65cm}

As illustrated in Figure~\ref{fig:localillu}, it is likely that the true change point falls within $\omega$ observations from the time the alarm is raised. This situation is the most frequent when the jump size is sufficiently large, as corroborated by Corollary~\ref{coro1}. The refinement algorithm aims to reduce localization error in such cases. Another scenario occurs when a false alarm is raised. As the true change point has not been reached in this situation, the refinement process has a probability of not triggering any alarms within the data window. In such cases, we can consider the refinement process as a confirmation step. If no alarm is raised during the refinement process, we can conclude that the previous alarm was a false alarm. We can then ignore it and continue running the algorithm. This approach reduces the probability of raising false alarms, which is particularly valuable when false alarms are costly in practical applications. Simulation D in Section~\ref{sim:refine} empirically demonstrates the effectiveness of the proposed refinement process. Details on the selection of $\omega'$ are provided in Section~\ref{sim:refine} of the supplementary material.


\vspace{-0.5cm}

\section{MULTIPLE CHANGE POINT SCENARIO}\label{multicp}
In this section, we consider the multiple change points case in which between change points, data is generated by stable and stationary VAR processes with different transition matrices and sub-Gaussian errors. Formally,  if we have a sequence of true change points $\{t^*_0 = 0, t^*_1, \ldots, t^*_{m-1}, t^*_{m} = T\}$, we have that
\begin{equation}\label{eqn:multicp}
    X_t = \sum_{l=1}^h A_l^{(j)} X_{t-l} + \varepsilon_t \,\,\, \text{for}\,\, t^*_{j-1} < t \leq t^*_{j},
\end{equation}
where, without loss of generality, we assume that the VAR processes between change points have the same order $h$ (otherwise, maximum of all lags will be selected as $h$) and the transition matrices between consecutive change points are different (i.e., $\{A_1^{(j)}, \ldots, A_h^{(j)}\} \neq \{A_1^{{(j+1)}}, \ldots, A_h^{{(j+1)}}\}$). Error terms $\varepsilon_t$ are independent zero mean sub-Gaussian random vectors with variance $\sigma^2_{(j)} I_p$ for $t^*_{j-1} < t \leq t^*_{j}$. Our Algorithm~\ref{alg1} can be implemented sequentially to address the proposed detection scenario, with the added assumption that the distance between change points is at least of the order \( s\left(\log (hp^2)\right) \). This requirement ensures a sufficient number of observations are available before the next change point, allowing accurate parameter estimation for monitoring. Previous theoretical results still holds under the same assumptions, if this new assumption is satisfied. This minimum distance condition is common in change point detection literature, see e.g. similar condition in Section 4.1 of \cite{safikhani2022joint} (see also \cite{safikhani2022fast}). The implementation of this sequential detection algorithm is briefly discussed as follows. Once a change point is detected, a new training period is initiated to estimate the new transition matrices and variances. Subsequently, the monitoring period will be based on these new estimations, as illustrated in Figure~\ref{fig:multiillu}. However, false alarms can be particularly costly under this implementation since they may trigger a training period that contains a change point. This situation not only leads to a missed detection but also contaminates the estimations for the upcoming monitoring. To effectively address this issue and reduce the probability of false alarms, we suggest applying the confirmation step, which was introduced in the refinement process in Section~\ref{refine}. It is also recommended to choose a conservative
(small) value for $\alpha$. The satisfactory performance of the sequential detection algorithm is confirmed empirically through synthetic data in Section~\ref{sim:mcp}.
\begin{figure}[!ht]
\begin{center}
\includegraphics[width=0.8\textwidth]{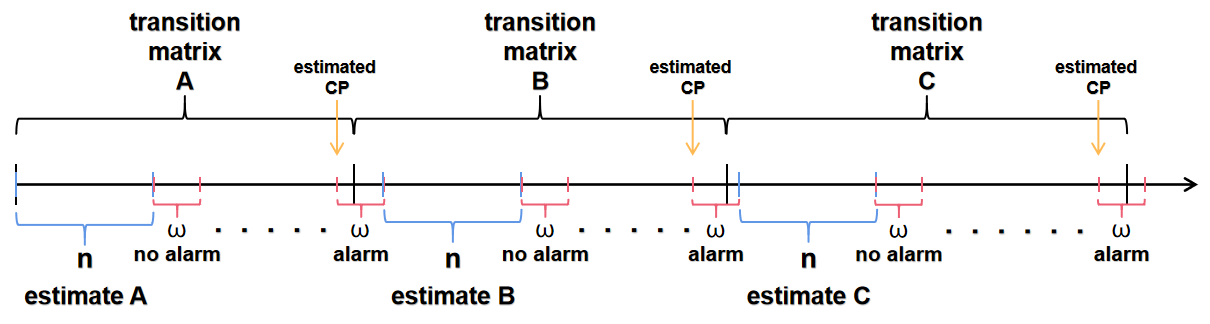}
\end{center}
\caption{Implementation of detection algorithm with multiple change points 
}
\label{fig:multiillu}
\end{figure}

\section{NUMERICAL COMPARISON}\label{section:num_main} Due to space constraints, the assessment of the proposed algorithm with simulated data is presented in Section~\ref{Appendix D} of the supplementary material. Section~\ref{Appendix D} includes analyses of average run length, detection delay, window size selection, refinement effectiveness, performance under multiple change point scenarios, and robustness to variance heterogeneity, time-varying transition matrices, and non-sparse transition matrices. In this section, we compare the empirical detection performance of our method with baseline methods: gstream, ocp, TSL, ocd, Mei, XS and Chan. The gstream method, proposed in \cite{chu2022sequential}, utilizes a k-nearest neighbor approach to sequentially detect change points. The implementation of this algorithm is provided by the authors in \cite{gStream}. The Bayesian online change point detection algorithm, proposed in \cite{adams2007bayesian}, is implemented in the R package ``ocp" \cite{Pagotto2019ocp}. The TSL algorithm, introduced in \cite{qiu2022transparent}, is a non-parametric approach for online change point detection in multivariate time series data. The algorithm is implemented in Fortran by the authors, and we use their provided function for threshold selection. The algorithms, namely ocd, Mei, XS, and Chan introduced in \cite{chen2022high,mei2010efficient, xie2013sequential, chan2017optimal}, respectively, are designed to detect changes in multivariate time series data observed sequentially. They utilize likelihood ratio tests in individual coordinates and aggregate the resulting statistics across scales and coordinates. These algorithms are available in the R package ``ocd." We calculate the threshold for each algorithm using Monte Carlo simulation, as outlined in Section 4.1 of \cite{chen2022high}.

\begin{figure*}[!ht]
\centering
\includegraphics[width=1\textwidth]{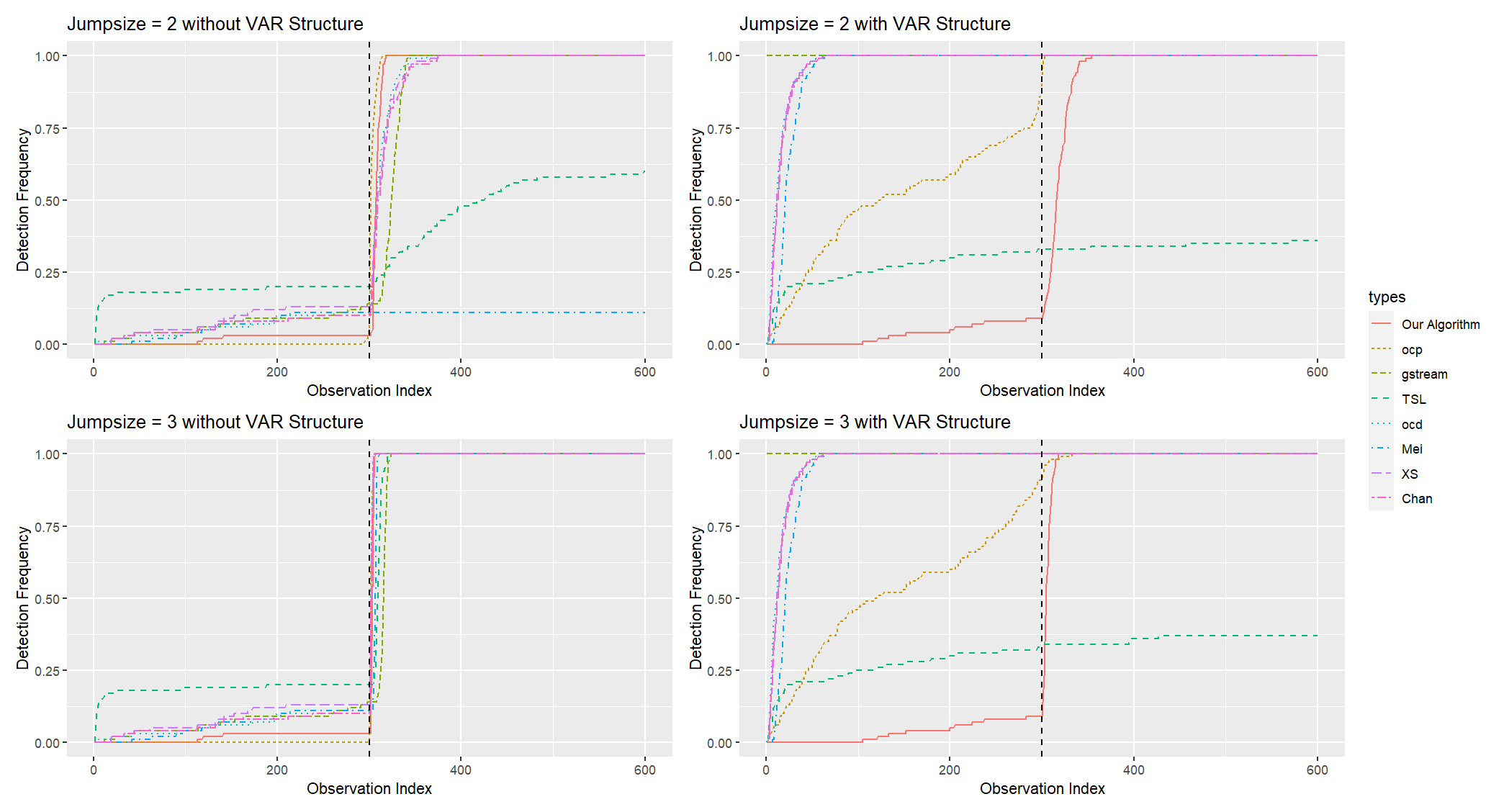}
\caption{Comparison: These plots provide a summary of the detection frequency for all algorithms. The black dashed vertical line represents the location of the true change point. A perfect algorithm would have a detection frequency of zero before the line and reach one immediately after the line.}\label{fig:compare}
\end{figure*}
\begin{table}[!ht]
\caption{Comparison: The average execution times for each algorithm are listed based on simulations conducted on data sets of varying lengths, where the initial $500$ data points are considered as the historical data set.}\label{tbl:exetime}
\begin{center}
\begin{tabular}{|c|c|c|c|c|c|c|c|c|}
\hline
& our algorithm & ocp & gstream & TSL & ocd & Mei & XS & Chan \\ \hline
Length=3500 & 0.11s & 18.26s & 141.59s & 46.58s & 1.84s & 0.68s & 1.34s & 1.33s \\ \hline
Length=4500 & 0.12s & 31.63s & 189.47s & 66.28s & 2.48s & 0.86s & 1.75s & 1.74s \\ \hline
Length=5500 & 0.14s & 48.06s & 237.06s & 90.95s & 3.08s & 1.11s & 2.18s & 2.21s \\ \hline
\end{tabular}
\end{center}
\end{table}

In the comparison, we apply our algorithm with $n = 500$, $\omega = 50$, and $\alpha = 1/1000$, along with the refinement and confirmation processes using a refine size of $1/10$. To ensure a fair comparison, we align most of the hyperparameters for the baseline methods with our choices. For example, we set $L = N0 = 500$ for gstream to match our choice of $n$, and we use $ARL = 1000$ and $\alpha = 1/1000$ for gstream. Similarly, we set $\lambda = 1000$ for ocp and $targetARL = 1000$ for TSL to match our choice of $\alpha$. The target average run length is also set to 1000 for ocd, Mei, XS, and Chan. For the remaining hyperparameters, we either use the recommendations by authors or perform a grid search under the same experimental settings, selecting the hyperparameters that yield the best performance. 

Since baseline algorithms are not specifically designed for data generated by a VAR process, we consider two scenarios in order to ensure a fair comparison. In the first scenario, data before the change point is generated from a constant mean model with independent and identically distributed errors (i.e., the transition matrix used is a zero matrix). In the second scenario, data before the change point is generated from a VAR process with a transition matrix of $0.8 \times I_p$. In both scenarios, the data after the change point is generated by a new VAR process with a different transition matrix. We vary the jump size between the old and new transition matrices to compare the algorithm's sensitivity. In each repetition, we generate a data sequence with a length of $1100$ and a dimension of $5$ (due to space constraints, the case with $p = 100$ is provided in Section~\ref{sim:hdcompare} of the supplementary material). The first $500$ observations are used as historical data, and a change point is located at $800$. After the first alarm is raised in each repetition, we terminate the algorithm and record the location of the last data point read at that time as $\tilde{t}$ (excluding historical data). We then construct an array consisting of $\tilde{t}$ zeros followed by $600 - \tilde{t}$ ones. This process is repeated for 100 repetitions, and we average the resulting arrays to obtain the detection frequency for each algorithm. A desirable algorithm should have a detection frequency of zero before $300$ and a detection frequency of one after $300$. The results are included in the Figure~\ref{fig:compare}.

In the left figures, when data is generated without a VAR structure, our algorithm performs comparably to the ocp algorithm and outperforms the other baseline methods. Our algorithm maintains a low false alarm rate (small detection frequency before $300$) and quickly detects the change (detection frequency reaches one shortly after $300$). However, when we add the VAR structure to the data, all baseline methods suffer from correlations and are unable to maintain the same low false alarm rate as before, as shown in the right figures. In contrast, our algorithm maintains similar performance as before, with only a slight increase in false alarm rate and detection delay.

Furthermore, we perform a brief simulation to assess the execution times of each algorithm. The setup and hyperparameter choices remain consistent with those in the comparison. We ensure that all algorithms continuously monitor the entire dataset with various sizes without halting even when an alarm is triggered, and we record the execution times. This procedure is repeated ten times to calculate the average execution times, and the outcomes are summarized in Table~\ref{tbl:exetime}.

\vspace{-0.5cm}

\section{REAL DATA EXPERIMENTS}\label{realdata}
We evaluate the effectiveness of our approach (VAR\_cpDetect\_Online) and contrast its performance with competing methods in two real-world scenarios: S\&P 500 data and EEG data. TSL, ocd, Mei, XS, and Chan are not suitable for this experimental setups as attempting to use them yielded unsatisfactory results, so we have excluded their outcomes from this section. The results for EEG data are deferred to Section~\ref{eeg} in the supplementary material to save space.

\vspace{-0.5cm}

\subsection{S\&P 500 Data}  
This dataset consists of adjusted daily closing prices for 186 stocks in the S\&P 500 from 2004-02-06 to 2016-03-02. Since closing prices are non-stationary, we apply the data cleansing approach from \cite{keshavarz2020sequential} to compute daily log returns for each stock, yielding a dataset with 186 columns (stocks) and 3037 rows (trading days). The first 200 data points (up to 2004-11-22) are designated as historical data. For our method, we set \(\omega = 22\) to match the typical number of trading days in a month and \(\alpha = 1/5000\). Hyperparameters for competing algorithms were selected similarly, as described in Section~\ref{section:num_main}. All methods were applied continuously to the entire dataset, without pausing when alarms were triggered. The locations of triggered alarms are documented in Section~\ref{addsp500} of the supplementary material to save space. The aim of this experiment is to detect abnormal states indicated by data deviations from the baseline period. Applying the multiple change point detection approach discussed in Section~\ref{multicp} is unsuitable here, as the distance between change points does not meet required assumptions. Instead, we estimate the starting points of alarm clusters. Consecutive alarms within a window of $\omega$ observations are grouped into the same cluster, as they often stem from overlapping data segments. This suggests a high probability of a shared underlying abnormality. For each cluster, the onset is estimated using the refinement method in Section~\ref{refine}. Specifically, if an alarm's distance from the previous alarm exceeds \(\omega\), it is treated as a new cluster's start, and the refinement procedure is used to estimate the change point. If not, the alarm is considered part of the existing cluster, and change point estimation is not performed. The proposed algorithm raised alarms that formed 13 distinct clusters, with the estimated starting points of these clusters treated as change points, as shown in Table~\ref{tbl:sp500}. This table also highlights the historical events likely influencing the S\&P500 index during these periods. In a previous study \citep{keshavarz2020sequential}, four alarm clusters were identified. Our algorithm aligns with these findings, with the starting points of these clusters indicated by an asterisk (*) in the table. Additionally, our algorithm identifies further periods of deviation from the baseline in the S\&P 500 index. The average detection delay for our algorithm is 10.6, below the pre-specified threshold of $\omega = 22$, meaning the algorithm triggers alarms after approximately 10 additional observations to mark the start of each cluster. The ocp method missed the change points identified in the previous study in October 2007 and December 2010, while the gstream method missed the change point in August 2014. Moreover, the gstream method raised excessive alarms, covering about 54.7\% of trading days, which limits its practicality for monitoring abnormal behavior. The execution times for the experiments were 50.63 seconds for our method, 227.64 seconds for ocp, and 50.03 seconds for gstream.

\begin{table}[!ht]
\centering
\begin{tabular}{|c|l|}
\hline
\textbf{Date}       & \textbf{Possible Real-World Event}                                                                     \\ \hline
2005-10-17          & Concerns about the housing bubble and economic slowdown.                                       \\ \hline
2006-07-13          & Fed hints at pausing rate hikes amid inflation and housing worries.                            \\ \hline
2007-07-19          & Early signs of the subprime mortgage crisis.                                                   \\ \hline
2007-10-12*          & U.S. housing market declines significantly.                                                    \\ \hline
2010-01-14          & Concerns over slow recovery and Eurozone debt crisis.                                           \\ \hline
2010-04-20          & BP oil spill raises environmental and economic concerns.                                        \\ \hline
2010-12-22*          & Strong holiday sales, but Eurozone concerns linger.                                             \\ \hline
2011-07-26*          & U.S. debt ceiling crisis and potential government default.                                      \\ \hline
2012-05-24          & Eurozone debt crisis, fears of Greece exiting the euro.                                         \\ \hline
2013-12-20          & Fed announces tapering of bond-buying program.                                                  \\ \hline
2014-08-21*          & Market turbulence from geopolitical tensions and growth concerns.                               \\ \hline
2015-08-14          & China devalues its currency, sparking global slowdown fears.                                    \\ \hline
2015-12-10          & Rising volatility ahead of expected Fed rate hike.                                              \\ \hline
\end{tabular}
\caption{Real-world events corresponding to changes in the S\&P500 index.}\label{tbl:sp500}
\end{table}

\section{CONCLUDING REMARKS}\label{future}
In this paper, we introduced an online change point detection algorithm specifically designed for high-dimensional VAR models. This algorithm can effectively detect changes in higher-order structures, such as cross correlations. The algorithm's test statistic utilizes one-step-ahead prediction errors over a moving window of data. We demonstrated the asymptotic normality of our proposed test statistic under relatively mild conditions, which can encompass high-dimensional scenarios where the number of parameters exceeds the sample size. Furthermore, we showed that the test's power approaches one with an increase in the jump size, and this was corroborated through numerical experiments. With respect to time complexity, we empirically demonstrated that our algorithm has a shorter computation time compared to competing algorithms. Our algorithm is currently tailored for data generated by VAR models with independent errors. Expanding its applicability to diverse error term structures is a challenging avenue for future research, given potential identifiability issues arising from general covariance structures. Further exploration includes relaxing some assumptions, such as substituting the sub-Gaussian distribution assumption on error terms with a more heavy-tail distribution assumption such as the sub-Weibull distribution. Investigating alternative forms of transition matrices, such as low rank or low rank combined with sparse structures \cite{basu2019low,bai2023multiple}, is another promising area for research. Integrating the sequential updating technique (briefly discussed in Section~\ref{app:sequpdate} of the supplementary material) to improve transition matrix estimation when no alarm has been raised is an important yet challenging area for future research.

\bibhang=1.7pc
\bibsep=2pt
\fontsize{9}{14pt plus.8pt minus .6pt}\selectfont
\renewcommand\bibname{\large \bf References}
\expandafter\ifx\csname
natexlab\endcsname\relax\def\natexlab#1{#1}\fi
\expandafter\ifx\csname url\endcsname\relax
  \def\url#1{\texttt{#1}}\fi
\expandafter\ifx\csname urlprefix\endcsname\relax\def\urlprefix{URL}\fi

\bibliographystyle{chicago}      
\bibliography{bib}   
 \newpage
 This supplementary material provides proofs of lemmas and theorems, as well as additional simulations and real data experiments. Definitions and lemmas used in the proofs are introduced in Sections~\ref{Appendix A} and~\ref{Appendix B}, respectively. The proofs of the theorems appear in Section~\ref{Appendix C}. Additional simulations are presented in Section~\ref{Appendix D}, and further real data analysis is provided in Section~\ref{addreal}. The use of a sequential updating technique for transition matrix estimation is discussed in Section~\ref{app:sequpdate}, and post-change analysis is provided in Section~\ref{postchangeanalysis}. 
 
\appendix

\section{DEFINITIONS}\label{Appendix A}
In this appendix, we provide definitions for symbols and terms utilized throughout the appendices.
\begin{definition}\label{def1}
\emph{For any $\gamma > 0$, a random variable $X$ satisfies any of the following equivalent properties is called a sub-Weibull ($\gamma$) random variable:
\begin{enumerate}
    \item $\mathbb{P}(|X|>t) \leq 2 \exp \left\{-\left(t / K_{1}\right)^{\gamma}\right\} \quad \text{ for all } t \geq 0 $,
    \item $\left(\mathbb{E}|X|^{p}\right)^{1 / p} \leq K_{2} p^{1 / \gamma} \quad \text{ for all } p \geq 1 \wedge \gamma,$
    \item $\mathbb{E}\left[\exp \left(|X| / K_{3}\right)^{\gamma}\right] \leq 2 .$
\end{enumerate}
The constants $K_1$, $K_2$ and $K_3$ differ from each other at most by a constant depending only on $\gamma$. The sub-Gaussian random variable is a special case of a sub-Weibull random variable with $\gamma=2$. The sub-Weibull norm of X is defined as $\|X\|_{\psi_{\gamma}}:=\sup _{p \geq 1}\left(\mathbb{E}|X|^{p}\right)^{1 / p} p^{-1 / \gamma}.$ }
\end{definition}
Definition~\ref{def1} is a straightforward combination of Lemma 5 and Definition 3 in \cite{wong2020lasso}.

\begin{definition}\label{def2}
\emph{For any $\gamma > 0$, a random vector $X \in \mathbb{R}^{p}$ is said to be a sub-Weibull $(\gamma)$ random vector if all of its one-dimensional projections are sub-Weibull $(\gamma)$ random variables. We define the sub-Weibull $(\gamma)$ norm of a random vector as
\[
\|X\|_{\psi_{\gamma}}:=\sup _{v \in S^{p-1}}\left\|v^{\prime} X\right\|_{\psi_{\gamma}},
\]
where $S^{p-1}$ is the unit sphere in $\mathbb{R}^{p}$.}
\end{definition}
Definition~\ref{def2} is from Definition 4 in \cite{wong2020lasso}.

\begin{definition}\label{def3}
\emph{Every VAR(h) process can be rewritten into a VAR(1) form that is $\tilde{X}_{t}=\tilde{A} \tilde{X}_{t-1}+\tilde{\varepsilon}_{t}$, and $\tilde{X}_t$ is stable if and only if $X_t$ is stable.}
\end{definition}

Definition~\ref{def3} is from \cite{lutkepohl2005new}. We omit the details here to save space; for more information, please refer to page 15 in \cite{lutkepohl2005new}.

\section{LEMMAS AND PROOFS}\label{Appendix B}
In this appendix, we introduce several lemmas that will be employed in the proof of the Theorems. We introduce the notation $D^l \coloneqq \hat{A}_l - A_l$ and define $d_{ij}^l$ as the (i, j)-th element of $D^l$. We utilize the symbol $N$ to represent a generic sample size, as opposed to exclusively using $n$ or $\omega$ in the subsequent lemmas. These lemmas will be applied with $N = n$ or $\omega$ during the proof of the Theorems.

\begin{lemma}\label{lemma1}
\emph{Consider a random realization $\left\{X_{-h+1}, \ldots, X_{N}\right\}$ generated from a $\operatorname{VAR(h)}$ process with Assumption~\ref{assumption1}-\ref{assumption4} satisfied. Then, there exist constants $c_{i}>0$ such that for all $N \succsim \max \left\{\nu_{LB}^{2}, 1\right\} s(2\log p+\log h)$, with probability at least $1-c_{1} \exp \left(-c_{2} N \min \left\{\nu_{LB}^{-2}, 1\right\}\right)$, 
\[
\theta^\prime\hat{\Gamma}_N\theta \geq \alpha_{LB}\|\theta\|^2_2 - \tau_{LB}^N\|\theta\|^2_1, \;\;\text{ for all }\theta\in\mathbb{R}^{hp^2}
\]
where
\[
\begin{aligned}
&\hat{\Gamma}_N\coloneqq I_{p} \otimes\left(\mathcal{X}_N^{\prime} \mathcal{X}_N / N\right),\nu_{LB}=c_{3} \frac{\mu_{\max }(\mathcal{A})}{\mu_{\min }(\mathcal{\tilde{A}})}, \alpha_{LB}=\frac{\sigma^2}{2 \mu_{\max }(\mathcal{A})}, \\
&\tau_{LB}^N=\alpha_{LB} \max \left\{\nu_{LB}^{2}, 1\right\} \frac{\log h+2\log p}{N},\\
&\mu_{\min }(\mathcal{A}):=\min _{|z|=1} \Lambda_{\min }\left(\mathcal{A}^{*}(z) \mathcal{A}(z)\right), \mu_{\max }(\mathcal{A}):=\max _{|z|=1} \Lambda_{\max }\left(\mathcal{A}^{*}(z) \mathcal{A}(z)\right),\\
&\mathcal{A}(z):=I_{p}- \sum_{l=1}^hA_lz^l \;and\;
\mathcal{\tilde{A}}(z):=I_{hp}- \tilde{A}z.
\end{aligned}
\]}
\end{lemma}
\begin{proof}
This lemma results from a straightforward application of Proposition 4.2 in \cite{basu2015regularized} to a VAR(h) process.
\end{proof}
\begin{lemma}\label{lemma2}
\emph{Consider a random realization $\left\{X_{-h+1}, \ldots, X_{N}\right\}$ generated from a $\operatorname{VAR(h)}$ process with Assumption~\ref{assumption1}-\ref{assumption4} satisfied. Then, there exist constants $c_{i}>0$ (different from Lemma \ref{lemma1}) such that for all $N \succsim \max\left\{\nu^2_{UB},1\right\}s\left(2\log p+\log h\right)$, with probability at least $1-c_{1} \exp \left(-c_{2} N\min\{\nu^{-2}_{UB},1\}\right)$, 
\[
\theta^\prime\hat{\Gamma}_N\theta \leq 3\alpha_{UB}\|\theta\|^2_2 + \tau_{UB}^N\|\theta\|^2_1, \;\;\text{ for all }\theta\in\mathbb{R}^{hp^2}
\]
where
$$
\begin{aligned}
   \alpha_{UB} &= \frac{\sigma^2}{2\mu_{\min }(\mathcal{A})},\;\nu_{UB}=54\frac{\mu_{\min }(\mathcal{A})}{\mu_{\min }(\mathcal{\tilde{A}})} \text{ and }\tau_{UB}^N = c_3\alpha_{UB}\max\left\{\nu^2_{UB},1\right\}\frac{\log h +2\log p }{N}. 
\end{aligned}
$$}
\end{lemma}
\begin{proof}
This proof closely resembles the proof of Proposition 4.2 in Appendix B of \cite{basu2015regularized}, so we will only highlight the necessary modifications. The unmentioned portions should adhere to the proof provided in \cite{basu2015regularized}. 

At the beginning of the proof in \cite{basu2015regularized}, besides $ \Lambda_{\min }\left(\Gamma_{\tilde{X}}(0)\right) \geq \frac{\Lambda_{\min }\left(\Sigma_{\epsilon}\right)}{\mu_{\max }(\mathcal{A})}$,
we also have $\Lambda_{\max }\left(\Gamma_{\tilde{X}}(0)\right) \leq \frac{\Lambda_{\max }\left(\Sigma_{\epsilon}\right)}{\mu_{\min }(\mathcal{A})}$ from Proposition 2.3 and the bounds in (4.1) in \cite{basu2015regularized}. Before applying Lemma 12, we set $\eta = \nu_{UB}^{-1}$ instead of $\omega^{-1}$. Then, applying the Lemma 12 in \cite{loh2012high} with $\delta=\Lambda_{\max }\left(\Sigma_{\epsilon}\right) / 54 \mu_{\min }(\mathcal{A}) \text { and } \Gamma=S-\Gamma_{\tilde{X}}(0)$ where S = $\left(\mathcal{X}_N^{\prime} \mathcal{X}_N / N\right)$, we have $\theta^\prime S \theta - \theta^\prime\Gamma_{\tilde{X}}(0)\theta  \leq\alpha_{UB}(\left\|\theta\right\|^2_2 + \frac{1}{k}\left\|\theta\right\|_1^2)$, so $\theta^\prime S \theta \leq 3 \alpha_{UB}\|\theta\|^2_2 + \frac{\alpha_{UB}}{k}\|\theta\|^2_1$ for all $\theta\in \mathbb{R}^{hp}$ with probability at least $1-2 \exp \left[-c N\min\{\nu^{-2}_{UB},1\} +2 k \log (d p)\right]$. Finally, we set k = $\left\lceil c N\min\{\nu^{-2}_{UB},1\}  /  4\log (h p)\right\rceil$ and follow the rest of proof in \cite{basu2015regularized} to get this Lemma. $\lceil x \rceil$ represents the smallest integer that is greater than or equal to $x$.
\end{proof}
To maintain symbol consistency, we use ``k" to denote the constant ``s" in \cite{basu2015regularized}, and ``s" represents the sparsity parameter ``k" in \cite{basu2015regularized}. In this proof, $\Lambda_{\max }\left(\Sigma_{\epsilon}\right)$ and $\Lambda_{\min }\left(\Sigma_{\epsilon}\right)$ degenerate to $\sigma^2$ because of the variance structure of errors in our model.

\begin{lemma}\label{lemma3}
\emph{Under the same setup of Lemma \ref{lemma1}, there exist constants $c_{i}>0$ such that for $N \succsim \log h + 2\log p$, with probability at least $1-c_{1} \exp \left[-c_{2}(\log h + 2 \log p)\right]$, we have
\[
\left\|\hat{\gamma}_N-\hat{\Gamma}_N \beta^{*}\right\|_{\infty} \leq \mathbb{Q}\left(\beta^{*}, \sigma^2\right) \sqrt{\frac{\log h + 2 \log p}{N}}
\]
where $\hat{\gamma}_N = (I \otimes \mathcal{X}_N^{\prime}) Y_N / N$ and $\mathbb{Q}(\beta^{*}, \sigma^2) =c_{0}[\sigma^2+\frac{\sigma^2}{\mu_{\min }(\mathcal{A})}+\frac{\sigma^2 \mu_{\max }(\mathcal{A})}{\mu_{\min }(\mathcal{A})}]$.
}
\end{lemma}
\begin{proof}
This lemma results from a straightforward application of Proposition 4.3 in \cite{basu2015regularized} to a VAR(h) process.
\end{proof}

\begin{lemma}\label{lemma4}
\emph{Consider the $\ell_1$ estimation problem (\ref{eqn:a3}) discussed in Section~\ref{estimation} in the main paper, under the same setup of Lemma \ref{lemma1}, with $N \geq 32\max \left\{\nu_{LB}^{2}, 1\right\} s(\log h + 2\log p)$. Then, there exist constants $c_{i}>0$ such that, for any \\ $\lambda_{N} \geq 4 \mathbb{Q}\left(\beta^{*}, \sigma^2\right) \sqrt{(\log h + 2 \log p) / N}$, any solution $\hat{\beta}_N$  of (3) satisfies
\[
\begin{aligned}
&\left\|\hat{\beta}_N-\beta^{*}\right\|_{1}  \leq 64 s \lambda_{N} / \alpha_{LB}, 
\left\|\hat{\beta}_N-\beta^{*}\right\|_{2}  \leq 16 \sqrt{s} \lambda_{N} / \alpha_{LB} \\
&\text{ and }\left(\hat{\beta}_N-\beta^{*}\right)^{\prime} \hat{\Gamma}_N\left(\hat{\beta}_N-\beta^{*}\right)  \leq 128 s \lambda_{N}^{2} / \alpha_{LB} 
\end{aligned}
\]
$\text{with probability at least }1-c_{1} \exp \left[-c_{2}(\log h + 2 \log p)\right]-c_{3} \exp \left(-c_{4} N \min \left\{\nu_{LB}^{-2}, 1\right\}\right)$.}
\end{lemma}
\begin{proof}
This lemma results from a straightforward application of Proposition 4.1, 4.2, and 4.3 in \cite{basu2015regularized} to a VAR(h) process, aided by the union bound.
\end{proof}

Here, we have summarized several equivalent expressions for the terms found in the aforementioned lemmas. We have
$$
\begin{aligned}
    \left\|\hat{\gamma}_N-\hat{\Gamma}_N \beta^{*}\right\|_{\infty} &=  \max_{j,j^\prime \in (1,\ldots,p)\atop l\in (1,\ldots,h)}\frac{1}{N}\left|\sum_{i=1}^N x_{i-l,j} \varepsilon_{i,j^\prime}\right|,\\
    \left\|\hat{\beta}_N-\beta^*\right\|_1 &= \sum_{l=1}^h \sum_{j=1}^p\sum_{j^\prime=1}^p\left|d_{j,j^\prime}^l\right| \text{ and}\\ 
    \left(\hat{\beta}_N-\beta^{*}\right)^{\prime} \hat{\Gamma}_N\left(\hat{\beta}_N-\beta^{*}\right) &= \frac{1}{N}\sum_{i=1}^N\sum_{j=1}^p\left(\sum_{l=1}^h\sum_{j^\prime=1}^pd_{jj^\prime}^l x_{i-l,j^\prime}\right)^2.
\end{aligned}
$$

\section{PROOF OF THEOREMS}\label{Appendix C}
In this appendix, we provide the proofs for Theorems \ref{th:1} and \ref{th:2} presented in the main paper.
\subsection*{Proof of Theorem \ref{th:1}:}
The proof of Theorem \ref{th:1} consists of two parts. The first part is the proof of $\sqrt{p\omega}(\frac{\hat{R}_t^{(n,\omega)}}{p}-\hat{\sigma}^2_n)\xrightarrow[]{D} \mathcal{N}(0,\operatorname{Var}(\varepsilon_{1,1}^2))$, and the second part is the proof of $\hat{V}_n\xrightarrow[]{p}\operatorname{Var}(\varepsilon^2_{1,1})$. Finally, applying Slutsky's Theorem leads to Theorem \ref{th:1}. By some straightforward algebra, we have
\[
\begin{aligned}
&\sqrt{p\omega}\left(\frac{\hat{R}_t^{(n,\omega)}}{p}-\hat{\sigma}^2_n\right) = \underbrace{\sqrt{p\omega}\left(\frac{1}{p\omega}\sum_{i = t + 1}^{t + \omega}\left\|\varepsilon_i\right\|^2_2 - \sigma^2\right)}_\text{term 1}\\
-& \underbrace{\sqrt{\frac{\omega}{n}}\sqrt{pn}\left(\frac{1}{pn}\sum_{i = 1}^{n}\|\varepsilon_i\|^2_2 - \sigma^2\right)}_\text{term 2}
- \underbrace{\frac{1}{\sqrt{p\omega}}\sum_{i = t + 1}^{t + \omega} \left\|\sum_{l=1}^h\left(\hat{A}_l - A_l\right)X_{i-l}\right\|^2_2}_\text{term 3}\\
-& \underbrace{\frac{2}{\sqrt{p\omega}}\sum_{i = t + 1}^{t + \omega}\sum_{l=1}^h X_{i-l}^\prime\left(\hat{A}_l-A_l\right)^\prime\varepsilon_i}_\text{term 4}
+ \underbrace{\frac{\sqrt{\omega}}{n\sqrt{p}}\sum_{i = 1}^{n} \left\|\sum_{l=1}^h\left(\hat{A}_l - A_l\right)X_{i-l}\right\|^2_2}_\text{term 5}\\
-& \underbrace{\frac{2\sqrt{\omega}}{n\sqrt{p}}\sum_{i = 1}^{n}\sum_{l=1}^h X_{i-l}^\prime\left(\hat{A}_l-A_l\right)^\prime\varepsilon_i}_\text{term 6}.
\end{aligned}
\]
\begin{itemize}
\item 
For term 1 and 2, because errors are iid with variance matrix $\sigma^2 I_p$, we have 
$$
\sqrt{p\omega}\left(\frac{1}{p\omega}\sum_{i = t + 1}^{t + \omega}\|\varepsilon_i\|^2_2 - \sigma^2\right)\xrightarrow[]{D} \mathcal{N}\left(0,\operatorname{Var}\left(\varepsilon_{1,1}^2\right)\right) \text{ as } \omega\xrightarrow[]{}\infty\text{, and}
$$
$$
\sqrt{\frac{\omega}{n}}\sqrt{pn}\left(\frac{1}{pn}\sum_{i = 1}^{n}\left\|\varepsilon_i\right\|^2_2 - \sigma^2\right) \xrightarrow[]{p} 0 \text{ as } n\xrightarrow[]{}\infty
$$ 
by Central Limit Theorem \cite{montgomery2010applied} and Slutsky's Theorem \cite{van2000asymptotic} under the condition $\omega = o(n)$.

\item
Under the condition for Theorem~\ref{th:1}, by the union bound in probability, with probability at least $1-c_{1} \exp (-c_{2} \omega\min \left\{\nu_{UB}^{-2}, 1\right\})-c_{3} \exp\left(-c_{4} n \min \left\{\nu_{LB}^{-2}, 1\right\}\right)-c_{5} \exp [-c_{6}(2 \log p + \log h)]$,
We have 
\[
\begin{aligned}
&\left(\hat{\beta}_n-\beta^{*}\right)^{\prime} \hat{\Gamma}_\omega\left(\hat{\beta}_n-\beta^{*}\right)\leq 3\alpha_{UB}\left\|\hat{\beta}_n-\beta^{*}\right\|^2_2 + \tau_{UB}^\omega\left\|\hat{\beta}_n-\beta^{*}\right\|_1^2\\
&\leq \left(3\alpha_{UB} + c \alpha_{UB}\max \left\{\nu_{UB}^{2},1\right\}s\frac{\log h + 2\log p}{\omega}\right) \left\|\hat{\beta}_n-\beta^{*}\right\|^2_2\\
&\leq \tilde{c} \frac{\alpha_{UB}}{\alpha_{LB}^2}\max \left\{\nu_{UB}^{2},1\right\}\mathbb{Q}^2\left(\beta^{*}, \sigma^2\right)s\frac{\log h + 2\log p}{n} \operatorname{max}\left(1, s\frac{(\log h + 2\log p)}{\omega}\right).
\end{aligned}
\]
The first inequality is by Lemma \ref{lemma2}; the second inequality comes from the fact that $\left\|\hat{\beta}_n-\beta^{*}\right\|_1 \leq 4\sqrt{s}\left\|\hat{\beta}_n-\beta^{*}\right\|_2$ which is proved in Appendix B: Proof of Proposition 4.1 in \cite{basu2015regularized}; for the last inequality, we apply Lemma \ref{lemma4} with $\lambda_n = 4\mathbb{Q}\left(\beta^{*}, \sigma^2\right) \sqrt{\left(\log h + 2 \log p\right) / n}$ and the fact that $(a+b)\leq2\operatorname{max}(a,b)$, where $\tilde{c}$ is a finite positive constant.
Hence, with $\omega=o(n)$ and $\sqrt{n}\succsim\frac{s(\log h + 2\log p)}{\sqrt{p}}$, term 3 converges to zero in probability as $n$ goes to infinity.
\item
For term 4, we have
\[
\begin{aligned}
\left|\text{term 4}\right| &= \left|\frac{2}{\sqrt{p\omega}}\sum_{l = 1}^{h}\sum_{j=1}^p\sum_{j^\prime=1}^p\left(d_{jj^\prime}^l\sum_{i = t+1}^{t+\omega}x_{i-l,j}\varepsilon_{i,j^\prime}\right)\right|\\
&\leq2\sqrt{\frac{\omega}{p}}\left(\sum_{l=1}^h\sum_{j = 1}^p\sum_{j^\prime=1}^p\left|d_{ij^\prime}^l\right|\right)\left(\max_{j,j^\prime \in \left(1,\ldots,p\right)\atop l\in \left(1,\ldots,h\right)}\frac{1}{\omega}\left|\sum_{i=t+1}^{t+\omega} x_{i-l,j} \varepsilon_{i,j^\prime}\right|\right)\\
&\leq \frac{c}{\alpha_{LB}}\mathbb{Q}^2\left(\beta^{*}, \sigma^2\right)\frac{s\left(\log h + 2\log p\right)}{\sqrt{np}}
\end{aligned}
\]
with probability at least $1-c_{1} \exp [-c_{2}(\log h + 2 \log p)] - c_{3} \exp (-c_{4} n \min \left\{\nu_{LB}^{-2}, 1\right\})$.
Conditions that $\omega \succsim\log h + 2 \log p$ and $n \geq 32 \max \left\{\nu_{LB}^{2}, 1\right\} s(\log h + 2\log p)$ are needed. The last inequality comes from the application of Lemma \ref{lemma3} on $\max_{j,j^\prime \in (1,\ldots,p)}\frac{1}{\omega}\left|\sum_{i=t+1}^{t+\omega} x_{i-1,j} \varepsilon_{i,j^\prime}\right|$ and the application of Lemma \ref{lemma4} on $\sum_{j = 1}^p\sum_{j^\prime=1}^p|d_{ij^\prime}| $ by choosing $\lambda_n$ to be the smallest possible value. Then, we applied the union bound to get the result.
Hence, under condition $\frac{s(\log h + 2\log p)}{\sqrt{p}} = o(\sqrt{n})$, we have the absolute value of term 4 converges to zero in probability as $n$ goes to infinity.
\item
For term 5, we have
$$
\text{term 5}\leq c\frac{1}{\alpha_{LB}}\mathbb{Q}^2\left(\beta^{*}, \sigma^2\right)\sqrt{\frac{\omega}{n}} \frac{s\left(\log h + 2\log p\right) }{\sqrt{np}},
$$ with probability at least $1-c_{1} \exp [-c_{2}(\log h + 2 \log p)]-c_{3} \exp \left(-c_{4} n \min \left\{\nu_{LB}^{-2}, 1\right\}\right)$,
under condition $n \geq \max 32\left\{\nu_{LB}^{2}, 1\right\} s(\log h + 2 \log p)$, by directly applying Lemma \ref{lemma4}. Hence, under condition $\omega = o(n)$ and $\sqrt{n}\succsim\frac{s(\log h + 2\log p)}{\sqrt{p}}$, we have term 5 converges to zero in probability as $n$ goes to infinity.
\item
For term 6, we have
\[
\begin{aligned}
 \left|\text{term 6}\right| &\leq  2\sqrt{\frac{\omega}{p}}\left(\sum_{l=1}^h\sum_{j = 1}^p\sum_{j^\prime=1}^p\left|d_{ij^\prime}^l\right|\right)*\left(\max_{j,j^\prime \in \left(1,\ldots,p\right)\atop l\in \left(1,\ldots,h\right)}\frac{1}{n}\left|\sum_{i=1}^{n} x_{i-l,j} \varepsilon_{i,j^\prime}\right|\right)\\
&\leq c \frac{1}{\alpha_{LB}}\mathbb{Q}^2\left(\beta^{*}, \sigma^2\right)\sqrt{\frac{\omega}{n}}\frac{s\left(\log h + 2 \log p\right)}{\sqrt{np}}
\end{aligned}
\]
with probability at least $1-c_{1} \exp [-c_{2}(\log h + 2 \log p)]-c_{3} \exp \left(-c_{4} n \min \left\{\nu_{LB}^{-2}, 1\right\}\right)$, under condition $n \geq 32\max \left\{\nu_{LB}^{2}, 1\right\} s(\log h + 2 \log p).$ Hence, under additional condition $\omega = o(n)$ and $\sqrt{n}\succsim\frac{s(\log h + 2\log p)}{\sqrt{p}}$, the absolute value of term 6 converges to zero in probability as $n$ goes to infinity.
\end{itemize}
Finally, by applying Slutsky's Theorem, we conclude the proof of the first part. In order to prove the main theorem, we still need to prove the second part: $\hat{V}_n\xrightarrow[]{p}\operatorname{Var}(\varepsilon^2_{1,1})$.
Firstly, note that 
\[
\begin{aligned}
\hat{\sigma}^2_n &= \underbrace{\frac{1}{pn}\sum_{i = 1}^{n} \left\|\sum_{l=1}^h\left(\hat{A}_l - A_l\right)X_{i-l}\right\|^2_2}_\text{term 1} - \underbrace{\frac{2}{pn}\sum_{i = 1}^{n}\sum_{l=1}^h X_{i-l}^\prime\left(\hat{A}_l-A_l\right)^\prime\varepsilon_i}_\text{term 2}
+ \underbrace{\frac{1}{pn}\sum_{i = 1}^{n}\left\|\varepsilon_i\right\|^2_2}_\text{term 3}.
\end{aligned}
\]
\begin{itemize}
\item 
For term 1, under the same condition in the first part of the proof, similar to the term 5 in the first part, we have
term 1 converges to zero in probability as $n$ goes to infinity by directly applying Lemma \ref{lemma4}.

\item
For term 2, under the same condition in the first part of the proof, similar to the term 6 in the first part, we have term 2 converges to zero in probability as $n$ goes to infinity by directly applying Lemmas \ref{lemma3} and \ref{lemma4}.
\item
For term 3, we have $\frac{1}{pn}\sum_{i = 1}^{n}\left\|\varepsilon_i\right\|^2_2 \xrightarrow[]{p}\operatorname{E}(\varepsilon^2_{1,1})\text{ as }n\xrightarrow[]{}\infty$ by applying the weak law of large number (\cite{ross2014first}), where E($x$) stands for the expectation of $x$. Thus, we have $\hat{\sigma}^4_n\xrightarrow[]{p}\operatorname{E}(\varepsilon^2_{1,1})^2$ by applying Slutsky's Theorem (\cite{van2000asymptotic}) and Continuous Mapping Theorem (\cite{billingsley2013convergence}).
\end{itemize}
On the other hand, we have
\[
\begin{aligned}
&\frac{1}{pn}\sum_{i=1}^n\left\|\left(\sum_{l=1}^h\hat{A}_l X_{i-l}\right)-X_i\right\|_4^4
= \underbrace{\frac{1}{pn}\sum^n_{i=1}\sum^p_{j=1}\left(\sum^p_{j^\prime=1}\sum^h_{l=1}d^l_{jj^\prime}x_{i-l,j^\prime}\right)^4}_{\text{term 1}}\\ &+ \underbrace{\frac{1}{pn}\sum^n_{i=1}\sum^p_{j=1}\varepsilon_{i,j}^4}_{\text{term 2}}
+ \underbrace{\frac{6}{pn}\sum^n_{i=1}\sum^p_{j=1}\left(\left(\sum^p_{j^\prime=1}\sum^h_{l=1}d^l_{jj^\prime}x_{i-l,j^\prime}\right)^2\varepsilon_{i,j}^2\right)}_{\text{term 3}}\\
&- \underbrace{\frac{4}{pn}\sum^n_{i=1}\sum^p_{j=1}\left(\left(\sum^p_{j^\prime=1}\sum^h_{l=1}d^l_{jj^\prime}x_{i-l,j^\prime}\right)^3\varepsilon_{i,j}\right)}_{\text{term 4}}
- \underbrace{\frac{4}{pn}\sum^n_{i=1}\sum^p_{j=1}\left(\left(\sum^p_{j^\prime=1}\sum^h_{l=1}d^l_{jj^\prime}x_{i-l,j^\prime}\right)\varepsilon_{i,j}^3\right)}_{\text{term 5}}.
\end{aligned}
\]
\begin{itemize}
\item 
For term 1, under the same condition for the first part of the proof, we have
\[
\begin{aligned}
\text{term 1}
\leq& \frac{1}{pn}\left(\sum^n_{i=1}\sum^p_{j=1}\left(\sum^p_{j^\prime=1}\sum^h_{l=1}d^l_{jj^\prime}x_{i-l,j^\prime}\right)^2\right)^2
\leq \frac{1}{pn}c^2\frac{\mathbb{Q}^4\left(\beta^{*}, \sigma^2\right)}{\alpha^2_{LB}}s^2\left(\log h + 2\log p\right)^2
\end{aligned}
\]
with probability at least $1-c_{1} \exp [-c_{2}(\log h + 2 \log p)]-c_{3} \exp (-c_{4} n \min \left\{\nu_{LB}^{-2}, 1\right\})$ by directly applying Lemma \ref{lemma4}.
Thus, we have term 1 converges to zero in probability as $n$ goes to inifinity.
\item
We have term 2 converges to $\operatorname{E}(\varepsilon^4_{1,1})$ in probability as $n$ goes to infinity by applying the weak law of large number \cite{ross2014first}.
\item
For term 3, under the same condition for the first part of the proof, we have
\[
\begin{aligned}
\text{term 3} &\leq \max_{j \in \left(1,\ldots,p\right)\atop i\in \left(1,\ldots,n\right)}\left(\frac{\varepsilon_{i,j}^2}{\sqrt{pn}}\right)\frac{6}{\sqrt{pn}}\sum^n_{i=1}\sum^p_{j=1}\left(\sum^p_{j^\prime=1}\sum^h_{l=1}d^l_{jj^\prime}x_{i-l,j^\prime}\right)^2 \\
&\leq c\frac{\mathbb{Q}^2\left(\beta^{*}, \sigma^2\right)}{\alpha_{LB}}\frac{s\left(\log h +2\log p\right)}{\sqrt{pn}}
\end{aligned}
\]
with probability at least \\
$1 - 2\exp(-c_5\sqrt{pn}+\log (np))-c_{1} \exp [-c_{2}(\log h + 2 \log p)]-c_{3} \exp (-c_{4} n \min \left\{\nu_{LB}^{-2}, 1\right\})$.
The last inequality is from the fact that errors are independent and identically distributed sub-Gaussian random variables, so we have
\[
\begin{aligned}
\operatorname{P}(\max_{j \in (1,\ldots,p)\atop i\in (1,\ldots,n)}(\frac{\varepsilon_{i,j}^2}{\sqrt{pn}})>C) &\leq np\operatorname{P}(\varepsilon_{1,1}^2>C\sqrt{pn})\leq 2\exp(-c_5\sqrt{pn}+\log (np))
\end{aligned}
\]
by Definition~\ref{def1} in Section~\ref{Appendix A}, where $C$ and $c_5$ are some finite positive constants. Then, by directly applying Lemma \ref{lemma4} on the rest of the term 3 together with union bound, we can get the last inequality for term 3. Thus, we have term 3 converges to zero in probability as $n$ goes to infinity.
\item
Under the same condition for the first part of the proof, we have that the absolute value of term 4 is less than or equal to
\[
\begin{aligned}
&\frac{4}{pn}\sum^n_{i=1}\sum^p_{j=1}\left(\left(\sum^p_{j^\prime=1}\sum^h_{l=1}d^l_{jj^\prime}x_{i-l,j^\prime}\right)^2\left|\sum^p_{j^\prime=1}\sum^h_{l=1}d^l_{jj^\prime}x_{i-l,j^\prime}\varepsilon_{i,j}\right|\right)\\
&\leq \left(\frac{4}{n^{1-a}p^{1-a}}\sum^n_{i=1}\sum^p_{j=1}\left(\sum^p_{j^\prime=1}\sum^h_{l=1}d^l_{jj^\prime}x_{i-l,j^\prime}\right)^2\right) \\
&* \left(\max_{j,j^\prime \in \left(1,\ldots,p\right)\atop i\in \left(1,\ldots,n\right), l\in \left(1,\ldots,h\right)}\left|\frac{x_{i-l,j^\prime}\varepsilon_{i,j}}{n^a p^a}\right|\right)\left(\sum_{l=1}^h \sum_{j=1}^p\sum_{j^\prime=1}^p\left|d_{j,j^\prime}^l\right|\right)\\
&\leq  c\frac{s\left(\log h + 2\log p\right)}{\sqrt{np}}\sqrt{\frac{s\left(\log h + 2\log p\right)}{n}}\frac{\sqrt{s}}{n^{1/2-a}p^{1/2-a}}
\end{aligned}
\]
with probability at least $1 - 2\exp(-\tilde{c}p^an^a+\log (np^2h))-c_{1} \exp [-c_{2}(\log h + 2 \log p)]-c_{3} \exp (-c_{4} n \min \{\nu_{LB}^{-2}, 1\})$, where $\tilde{c}$ and $c$ are some finite positive constants and $a$ is an arbitrary small positive constant less than 1/2.
The last inequality is by applying the Lemma \ref{lemma4} on the first and last terms. For the middle term, we have for a positive finite constant $c^*$, $\operatorname{P}\left(\max_{j,j^\prime \in \left(1,\ldots,p\right)\atop i\in \left(1,\ldots,n\right), l\in \left(1,\ldots,h\right)}\left|\frac{x_{i-l,j^\prime}\varepsilon_{i,j}}{n^a p^a}\right|>c^*\right)
\leq \sum_{i,j,j^\prime,l}\operatorname{P}\left(|x_{i-l,j^\prime}\varepsilon_{i,j}|>p^an^ac^*\right)$. According to E.1 VAR section in \cite{wong2020lasso} with the assumption of stability and stationarity (Assumption~\ref{assumption3}), we know that $x_{i-l,j^\prime}$ is sub-Gaussian for all i, l and $j^\prime$. Then, according to Proposition 2.3 in \cite{vladimirova2020sub}, we have $x_{i-l,j^\prime}\varepsilon_{i,j}$ is sub-weibull ($\gamma=1$) for all i, j, $j^\prime$ and l. Thus, according to Definition \ref{def1}, we have for some finite positive constant $\tilde{c}$, $\sum_{i,j,j^\prime,l}\operatorname{P}\left(|x_{i-l,j^\prime}\varepsilon_{i,j}|>p^an^ac^*\right)\leq 2\exp\left(-\tilde{c}p^a n^a + \log (np^2h)\right)$.
Finally, by the union bound we can get the last inequality. Thus, we have under the same condition for the proof of first part with additional condition, $n^{1/2-a}\succsim\frac{\sqrt{s}}{p^{1/2-a}}$ for some $a \in (0,1/2)$, the absolute value of term 4 converges to zero in probability as $n$ goes to infinity.
\item
Under the same condition for the first part of the proof, the absolute value of term 5 is less than or equal to
\[
\begin{aligned}
& 4 \left(\max_{j,j^\prime \in \left(1,\ldots,p\right)\atop i\in \left(1,\ldots,n\right), l\in \left(1,\ldots,h\right)}\left|\frac{x_{i-l,j^\prime}\varepsilon^3_{i,j}}{n^b p^b}\right|\right)\frac{n^b}{p^{1-b}}\left(\sum_{l=1}^h \sum_{j=1}^p\sum_{j^\prime=1}^p\left|d_{j,j^\prime}^l\right|\right)\\
&\leq c\frac{\mathbb{Q}\left(\beta^{*}, \sigma^2\right)}{\alpha_{LB}}\frac{\sqrt{s\left(\log h + 2\log p\right)}}{n^{1/4}p^{1/4}}\frac{\sqrt{s}}{p^{3/4-b}n^{1/4-b}}
\end{aligned}
\]

with probability at least $1 - 2\exp(-\tilde{c}p^{b/2}n^{b/2}+\log (np^2h))-c_{1} \exp [-c_{2}(\log h + 2 \log p)]-c_{3} \exp (-c_{4} n \min \{\nu{LB}^{-2},1\})$, where $\tilde{c}$ and $c$ are some finite positive constants and $b$ is an arbitrary small positive constant less than $1/4$. Derivations of term 4 and 5 are similar with the only change that $x_{i-l,j^\prime}\varepsilon^3_{i,j}$ is sub-weibull($\gamma=1/2$) for all $i$, $j$, $j^\prime$ and $l$. Thus, we have under the same condition for the proof of first part with additional condition that $n^{1/4-b}\succsim\frac{\sqrt{s}}{p^{3/4-b}}$ for some $b \in (0,1/4)$, the absolute value of term 5 converges to zero in probability as $n$ goes to infinity.
\end{itemize}
Then, by applying Slutsky’s Theorem and Continuous Mapping Theorem, we conclude the proof of the second part. Finally, applying Slutsky's Theorem again yields Theorem \ref{th:1}.

\subsection*{Proof of Theorem \ref{th:2}:}
By some straightforward algebra, we have
$\hat{R}_{t^*+h}^{(n,\omega)}$ is equal to
$$
\begin{aligned}
 &\underbrace{\frac{1}{\omega}\sum_{i = {t^*+h} + 1}^{{t^*+h} + \omega}\left\|\left(\sum_{l=1}^h\left(\hat{A}_l - A_l\right)X_{i-l}\right) - \varepsilon_i \right\|^2_2}_{T_1} \underbrace{+\frac{1}{\omega}\sum_{i = {t^*+h} + 1}^{{t^*+h} + \omega}\left\|\sum_{l=1}^h\left(A_l - A_l^*\right)X_{i-l}\right\|^2_2}_{T_2}\\
&\underbrace{-\frac{2}{\omega}\sum_{i = {t^*+h} + 1}^{{t^*+h} + \omega}\left(\varepsilon_i^T\left(\sum_{l=1}^h\left(A_l - A_l^*\right)X_{i-l}\right)\right)}_{T_3}\\
&\underbrace{+ \frac{2}{\omega}\sum_{i = {t^*+h} + 1}^{{t^*+h} + \omega}\left(\left(\sum_{l=1}^h\left(\hat{A}_l - A_l\right)X_{i-l}\right)^T\left(\sum_{l=1}^h\left(A_l - A_l^*\right)X_{i-l}\right)\right)}_{T_4}.
\end{aligned}
$$
 Thus, we have that 
 $$
 \hat{T}_{t^*+h}^{(n,\omega)} = \underbrace{\sqrt{\frac{p\omega}{\hat{V}_n}}\left(\frac{T_1}{p}-\hat{\sigma}^2_n\right)}_{\text{term 1}} + \underbrace{\sqrt{\frac{p\omega}{\hat{V}_n}}\frac{T_2}{p}}_{\text{term 2}} + \underbrace{\sqrt{\frac{p\omega}{\hat{V}_n}}\frac{T_3}{p}}_{\text{term 3}} + \underbrace{\sqrt{\frac{p\omega}{\hat{V}_n}}\frac{T_4}{p}}_{term 4}.
$$

\begin{itemize}
    \item We have term 1 converges to $\mathcal{N}(0,1)$ in distribution as $n$ goes to infinity by the proof of Theorem~\ref{th:1}.
    
    \item For term 2, we have 
    $$
    \begin{aligned}
    \sqrt{\frac{p\omega}{\hat{V}_n}}\frac{T_2}{p} &= \frac{1}{\sqrt{\hat{V}_n}}\sqrt{\frac{\omega}{p}} \left(\beta^*-\beta_{new}\right)^{\prime} \hat{\Gamma}_\omega\left(\beta^*-\beta_{new}\right) \geq \frac{1}{\sqrt{\hat{V}_n}}\sqrt{\frac{\omega}{p}}(\alpha_{LB}^\prime - s\left(\tau^{\omega}_{LB}\right)^\prime)\left\|\beta^*-\beta_{new}\right\|^2_2\\
    &= \frac{\alpha_{LB}^\prime}{\sqrt{\hat{V}_n}}\sqrt{\frac{\omega}{p}}\left\|\beta^*-\beta_{new}\right\|^2_2 - c\frac{s(\log h + 2 \log p)}{\sqrt{\omega p}}\frac{1}{\sqrt{\hat{V}_n}}\left\|\beta^*-\beta_{new}\right\|^2_2
    \end{aligned} 
    $$
    with probability at least $1 - c_1 \exp(-c_2\omega \min((\nu_{LB}^\prime)^{-2},1))$ by using the sparsity assumption and Lemma~\ref{lemma1}. 

    On the other hand, by using the sparsity assumption and Lemma 2, we have 
    $$
    \begin{aligned}
    \sqrt{\frac{p\omega}{\hat{V}_n}}\frac{T_2}{p} 
    &\leq \frac{3\alpha_{UB}^\prime}{\sqrt{\hat{V}_n}}\sqrt{\frac{\omega}{p}}\left\|\beta^*-\beta_{new}\right\|^2_2 + c^\prime\frac{s(\log h + 2 \log p)}{\sqrt{\omega p}}\frac{1}{\sqrt{\hat{V}_n}}\left\|\beta^*-\beta_{new}\right\|^2_2,
    \end{aligned} 
    $$
    with probability at least $1 - c_3 \exp(-c_4\omega \min((\nu_{UB}^\prime)^{-2},1))$. Further, $c$ and $c^\prime$ are some positive constants; $\alpha_{LB}^\prime$, $\left(\tau^{\omega}_{LB}\right)^\prime$, $\nu_{LB}^\prime$, $\alpha_{UB}^\prime$, $\left(\tau^{\omega}_{UB}\right)^\prime$ and $\nu_{UB}^\prime$ refer to the corresponding components for the new VAR process after the change point. Finally, by the condition $s(\log h + 2 \log p) = o(\omega)$, we have $\sqrt{\frac{p\omega}{\hat{V}_n}}\frac{T_2}{p} \asymp \sqrt{\frac{\omega}{p}}\left\|\beta^*-\beta_{new}\right\|^2_2$.
    
    \item For term 3, we have 
    $$
    \begin{aligned}
    \left|\sqrt{\frac{p\omega}{\hat{V}_n}}\frac{T_3}{p}\right| &\leq  \frac{2}{\sqrt{\hat{V}_n}}\sqrt{\frac{\omega}{p}} \left\|\beta^*-\beta_{new}\right\|_1 * \left\|\hat{\gamma}_\omega-\hat{\Gamma}_\omega \beta_{new}\right\|_{\infty}\\
    &\leq 2\frac{\mathbb{Q}\left(\beta_{new}, \sigma^2\right)}{\sqrt{\hat{V}_n}} \sqrt{\frac{s(\log h + 2 \log p)}{p}}\left\|\beta^*-\beta_{new}\right\|_2
    \end{aligned}
    $$
    with probability at least $1-c_{5} \exp \left[-c_{6}(\log h + 2 \log p)\right]$ by sparsity assumption and Lemma~\ref{lemma3}, while $c_5$ and $c_6$ are some positive constant. With condition $\sqrt{\frac{s(\log h + 2 \log p)}{\omega}} = o(\|\beta^*-\beta_{new}\|_2)$, we have $|\sqrt{\frac{p\omega}{\hat{V}_n}}\frac{T_3}{p}| = o_p(\sqrt{\frac{\omega}{p}}\|\beta^*-\beta_{new}\|^2_2)$.
    
    \item We have that the absolute value of term 4 is less than or equal to 
    $$
    \begin{aligned}
    &\frac{2}{\sqrt{\hat{V}_n}}\sqrt{\frac{\omega}{p}}\omega^\eta p^\eta \max_{i,j^\prime,l}\left(\left|\frac{x_{i-l,j^\prime}^2}{\omega^\eta p^\eta}\right|\right)\left\|\hat{\beta}_n-\beta^*\right\|_1
    \left\|\beta^*-\beta_{new}\right\|_1\\
    &\leq \frac{c}{\sqrt{\hat{V}_n}}\omega^\eta p^\eta s \sqrt{\frac{s(\log h + 2 \log p)}{n}} \sqrt{\frac{\omega}{p}} \left\|\beta^*-\beta_{new}\right\|_2
    \end{aligned}
    $$
    with probability at least $1-c_{7} \exp [-c_{8}(\log h + 2 \log p)]-c_{9} \exp (-c_{10} n \min \{\nu_{LB}^{-2}, 1\}) - 2\exp (-c_{11}p^\eta \omega^\eta + \log (\omega p h))$ for some positive constant c and some $\eta \in (0,\frac{1}{4})$. To get the inequality above, $\max_{i,j^\prime,l}\left(\left|\frac{x_{i-l,j^\prime}^2}{\omega^\eta p^\eta}\right|\right)$ is bounded by a positive constant with high probability by using the properties of Sub-Weibull distribution and the nature of stationary time series. This part is very similar to the second part of the proof of the Theorem~\ref{th:1}. In addition, lemma~\ref{lemma4} is applied to bound $\left\|\hat{\beta}_n-\beta^*\right\|_1$. With condition $\omega^\eta p^\eta  \sqrt{\frac{s^3(\log h + 2 \log p)}{n}} = o(\|\beta^*-\beta_{new}\|_2)$, we have $|\sqrt{\frac{p\omega}{\hat{V}_n}}\frac{T_4}{p}| = o_p(\sqrt{\frac{\omega}{p}}\|\beta^*-\beta_{new}\|^2_2)$.
\end{itemize}
Combining these four terms, we will have the inequality in Theorem~\ref{th:2} with probability at least 1 - $\epsilon_{n,p,\omega}$ by the union bound. We have $L_{t^*+h}^{(n,\omega)} = o_p(\sqrt{\frac{\omega}{p}}\|\beta^*-\beta_{new}\|^2_2)$ and $(L_{t^*+h}^{(n,\omega)})^\prime = o_p(\sqrt{\frac{\omega}{p}}\|\beta^*-\beta_{new}\|^2_2)$, because we have (1) $\hat{V}_n \xrightarrow[]{p}\operatorname{Var}(\varepsilon^2_{1,1})$ as $n \xrightarrow[]{}\infty$ by the proof of Theorem \ref{th:1}, and (2) additional conditions for Theorem \ref{th:2}. This concludes the proof of Theorem \ref{th:2}.

\section{NUMERICAL STUDIES}\label{Appendix D}
In this section, we evaluate the performance of our algorithm using synthetic data generated by a VAR process. The primary metrics used for assessing the algorithm's effectiveness are the run length and detection delay, which are standard measures for online change point algorithms. These metrics have also been used in previous studies, such as \cite{chen2022high,mei2010efficient, xie2013sequential, chan2017optimal}. To compute the run length, we apply our algorithm to a data set without any change points and record the number of observations monitored before the first alarm is raised. On the other hand, to determine the detection delay, we apply our algorithm to a data set containing a change point and record the distance between the location of the last observation read by the algorithm and the true location of the change point after the alarm is correctly triggered. 

\subsection{Simulation A: Run Length}\label{sim:runlength}

The majority of online change point detection algorithms offer parameters that allow practitioners to control the target average run length (ARL). The target average run length represents the expected number of observations or time steps required by the algorithm to raise an alarm when applied to a data sequence without any actual change points. The ARL is an essential measure to balance the algorithm's performance between being sensitive enough to detect changes promptly and avoiding excessive false alarms. In our algorithm, the target average run length is primarily influenced by the choice of parameter $\alpha$. Specifically, setting $\alpha$ to be $1/1000$ will result in a lower bound of $1000$ for the ARL of our algorithm. For this simulation scenario, our focus is on exploring how to regulate the average run length by selecting an appropriate $\alpha$, as well as investigating how the dimension of data and the size of training data affect the run length of our algorithm. In this simulation, we consider three different choices for the parameter $\alpha$, namely $1/1000$, $1/5000$, and $1/10000$. We estimate transition matrices and variances using training data with sizes $n$ equal to $500$, $1000$, $1500$, and $2000$. Additionally, we vary the dimension of the data, setting it to be $10$, $40$, $70$, and $100$. For each combination of $\alpha$, $n$, and $p$, we generate $10/\alpha + n$ data points from a lag-$1$ VAR process without any change points. After estimating the transition matrices and variances using $n$ observations, we proceed to apply our algorithm on the remaining data. The algorithm is run with a pre-specified detection delay of $\omega$ set to $50$ and $h$ set to $1$. We repeat this process $200$ times, recording all the run lengths for each combination of parameters. The box plots of these run lengths are presented in Figure~\ref{fig:RLparameters}.

\begin{figure}[!ht]
\centering
\includegraphics[width=1\textwidth]{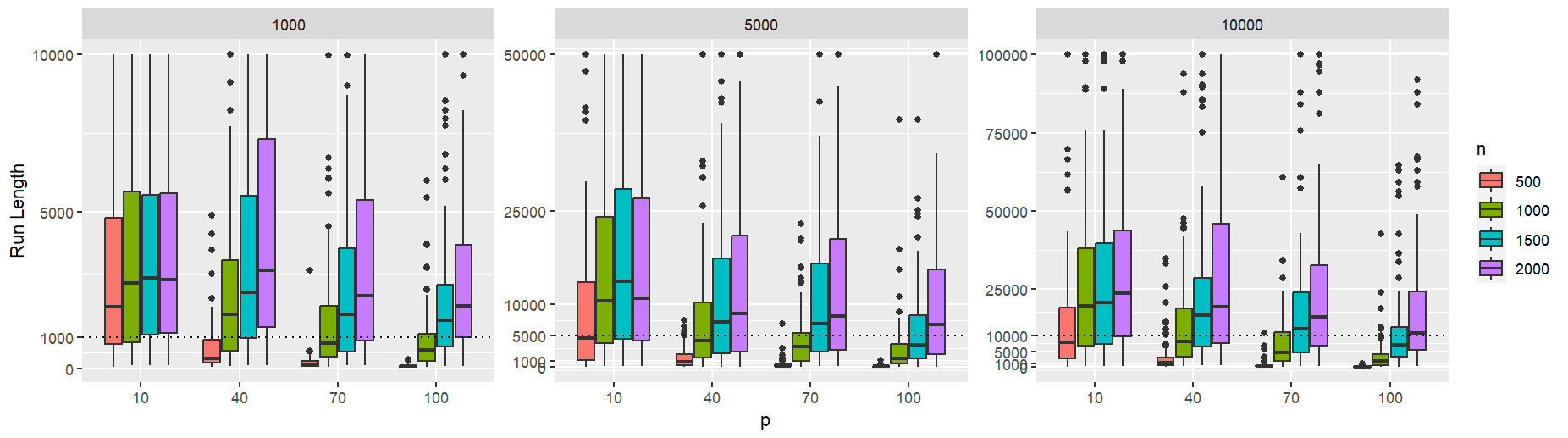}
\caption{Simulation A: This plot displays the box plots representing the run lengths of our algorithm for different combinations of training data size $n$, data dimension $p$, and $\alpha$. The values $1000$, $5000$, and $10000$ correspond to the target ARL, which is controlled by setting $\alpha$ to $1/1000$, $1/5000$, and $1/10000$, respectively. }
\label{fig:RLparameters}
\end{figure}

As depicted in the plot, when the training sample size is adequately large, the run lengths of our algorithm are consistently lower bounded by $1/\alpha$ with high probability. However, when training data is limited, the run lengths may not reach the target run length. Therefore, we recommend that practitioners set $1/\alpha$ to be equal to the length of the data that needs to be monitored when there is sufficient training data available. In cases where training data is limited, further decreasing the value of $\alpha$ might be a viable solution to reduce the probability of false alarms. Another noteworthy observation from this figure is that as the dimension of the data increases, the size of the training data set needed to maintain a satisfactory run length also increases.

\subsection{Simulation B: Detection Delay}\label{sim:detectiondelay}
\begin{figure}[!ht]
\centering
\includegraphics[width=1\textwidth]{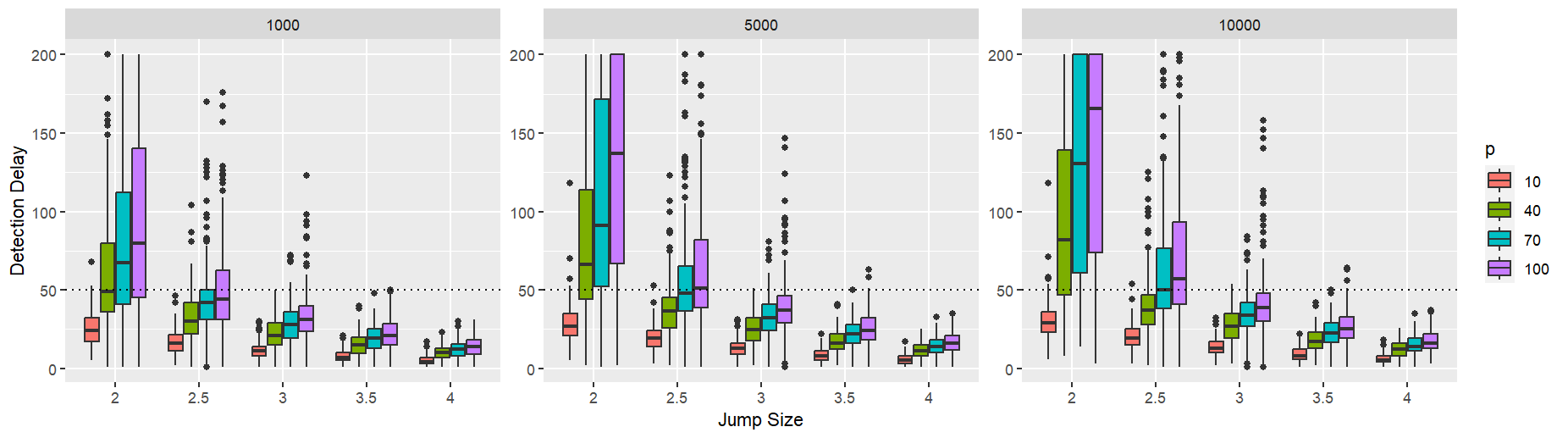}
\caption{Simulation B: This plot displays the box plots representing the detection delays of our algorithm for different combinations of jump size, data dimension $p$, and $\alpha$. The values $1000$, $5000$, and $10000$ correspond to the target ARL, which is controlled by setting $\alpha$ to $1/1000$, $1/5000$, and $1/10000$, respectively. The horizontal dashed line represents the pre-specified detection delay, denoted as $\omega$. }\label{fig:ddparameters}
\end{figure}

Detection delay measures the time lag between the occurrence of a change point and the moment the algorithm successfully detects it. A shorter detection delay implies that the algorithm can quickly identify and adapt to changes, which is critical in real-time systems where timely reactions are necessary to mitigate potential risks or capitalize on emerging opportunities. In this simulation, we explore how the detection delay of our algorithm is influenced by different choices of $\alpha$, data dimension $p$, and the jump size of the change point. Specifically, we set $\alpha$ to three different values: $1/1000$, $1/5000$, and $1/10000$, and vary the data dimension to $10$, $40$, $70$, and $100$, as well as the jump size to $2$, $2.5$, $3$, $3.5$, and $4$. To focus solely on the detection delay and eliminate the impact of false alarms, we generate data points from a lag-$1$ VAR process with a total length of $2200$. The change point is located at position $2000$, which corresponds to the end of the training period. By doing so, we can consider the number of observations our algorithm reads before raising an alarm as the detection delay. We run our algorithm with a pre-specified detection delay set to $50$ and recorded the corresponding detection delay. This process was repeated 200 times for each combination of parameters. The resulting detection delays were then summarized using box plots, as shown in Figure~\ref{fig:ddparameters}.

As depicted in the figure, when the jump size is large, the detection delay of our algorithm is consistently upper bounded by the pre-specified detection delay with high probability. This finding aligns with Corollary~\ref{coro1}, confirming that the detection delay will be upper bounded by $\omega + h$ with high probability when the jump size is sufficiently large. On the other hand, when we choose a smaller value for $\alpha$, the detection delay of our algorithm increases. Although this effect is only pronounced when the jump size is small, it is still essential to select an appropriate $\alpha$ to strike a balance between the detection delay and the probability of false alarms in practical applications. Another observation from the figure is that as the dimension of the data increases, a larger jump size is required to achieve a small detection delay. This observation aligns with our assumption on the jump size, as introduced in Theorem~\ref{th:2}.

\subsection{Simulation C: Choice of $\omega$}\label{sim:w}

\begin{figure}[!ht]
\centering
\includegraphics[width=1\textwidth]{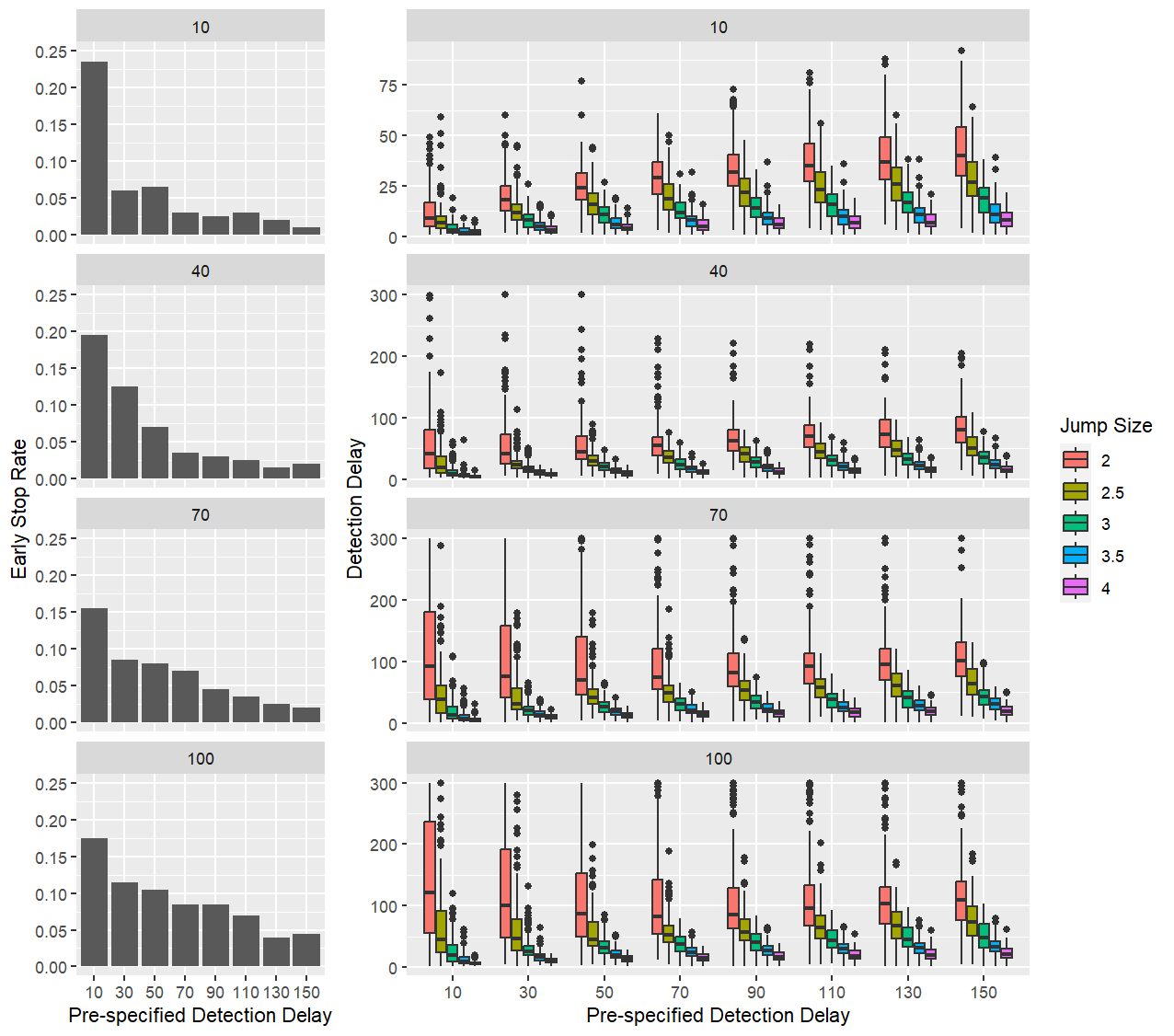}
\caption{Simulation C: The plot on the left summarizes the early stop rates, while the plot on the right presents the detection delays. For each grid in the plots, the dimension of data is set to $10$, $40$, $70$, and $100$.}\label{fig:windowsize}
\end{figure}

The pre-specified detection delay can be regarded as a moving window that contains data points used to compute the test statistic for our algorithm. The selection of its size, denoted as $\omega$, significantly impacts the performance of our algorithm in terms of both the probability of false alarms and the detection delay. Thus, in this simulation, our main objective is to investigate how various choices of $\omega$ influence the detection delay and early stop rate of our algorithm under different combinations of change point jump size and data dimension. For each combination of $p$, $\omega$, and jump size, we generate a data set of $2600$ data points using a lag-$1$ VAR process, with the change point occurring at time $2300$. We then estimate the transition matrices and variances using the first $2000$ data points and begin monitoring from that point onward. During the monitoring process, if our algorithm raises a false alarm before reaching the true change point, we consider it an early stop and record this occurrence. On the other hand, if the algorithm raises an alarm after the true change point, we record the detection delay. The $\alpha$ is set to $1/1000$ in all combinations. This entire process is repeated $200$ times. The early stop rate is calculated by dividing the number of early stops by $200$ and all detection delays are recorded for each combination. The results are presented and summarized in Figure~\ref{fig:windowsize}.

Figure~\ref{fig:windowsize} demonstrates that larger values of \( \omega \) are preferable for effectively controlling the false alarm rate. This is reflected in the early stop rate shown in the left panel. However, selecting \( \omega \) becomes more intricate when aiming to minimize detection delay, as it is highly sensitive to the jump size and the data dimensionality. In practice, this dependence makes it challenging to derive an optimal data-driven approach for selecting \( \omega \) when the true changes and jump sizes are unknown. According to the conditions specified in the theoretical results, \( \omega \) should scale as \( c \log(hp^2) \) for some constant \( c>0 \). Carefully reviewing the results in Figure~\ref{fig:windowsize}, for practical implementation, \( \omega = 10 \log(hp^2) \) is recommended. This choice results in \( \omega \) values of 46, 74, 85, and 92 for dimensions \( p = 10, 40, 70,\) and \(100\), respectively. These values effectively maintain a low early stop rate while minimizing detection delay for small jump sizes (e.g., jump size = 2). As shown in the figure, the impact of \( \omega \) on detection delay is more pronounced for smaller jump sizes. Although this choice may not yield the optimal delay for larger jump sizes, it incurs only a minor increase in detection delay relative to the optimal \( \omega \). However, when the training sample size is small, adjusting \( \omega \) has limited effect on enhancing detection quality. Moreover, when the estimation of transition matrices is imprecise, a larger window size can introduce more error into the test statistic, which aligns with the condition in Theorem \ref{th:1}, where \( \omega = o(n) \). In practical settings, a training sample size approximately 10–15 times the window size \( \omega \) is advisable, which can serve as a reference for selecting \( \omega \) when training data is limited.

\subsection{Simulation D: Effectiveness of Refinement}\label{sim:refine}

\begin{figure}[!ht]
\centering
\includegraphics[width=1\textwidth]{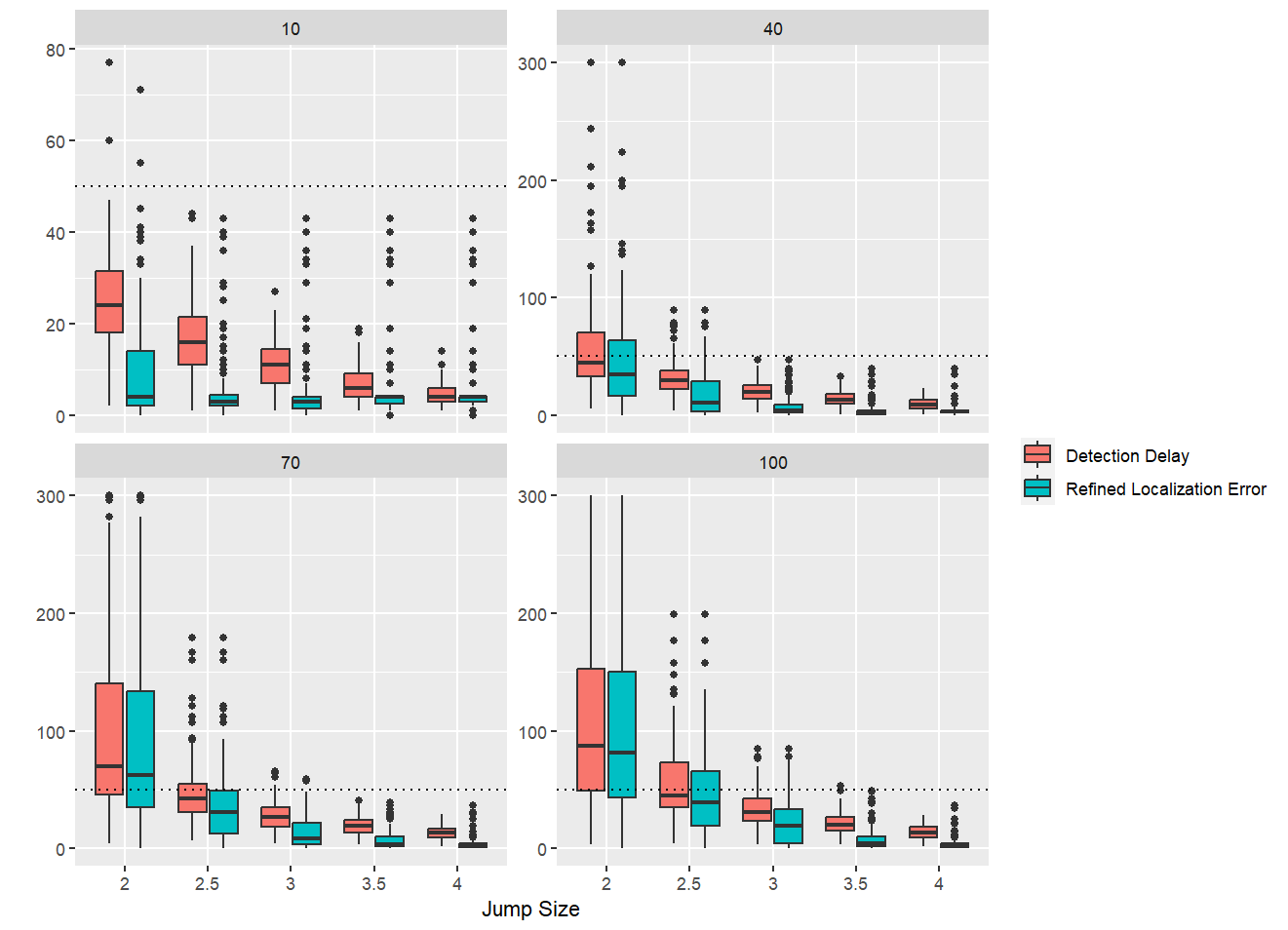}
\caption{Simulation D: This plot corresponds to box plots of detection delays and refined localization errors. For each grid in the plots, the dimension of data is set to $10$, $40$, $70$, and $100$. The horizontal dashed line represents the pre-specified detection delay, denoted as $\omega$.}\label{fig:rs01}
\end{figure}

In this simulation, our primary focus is to assess the effectiveness of the proposed change point localization refinement process, which was introduced in Section~\ref{refine}. Before delving into the simulation setup, we first introduce a few terms related to the refinement process. The first term is the ``refine size," which is defined as the ratio between the new pre-specified detection delay and the old pre-specified detection delay. For instance, if the refine size is set to $0.1$, and the original $\omega$ is $50$, then the value of $\omega^\prime$ used in the refinement process will be $5$. The second term is the ``refined localization error," representing the distance between the refined location of the estimated change point and the true location of the change point. Formally, if an alarm is raised at time $\hat{t}$ (i.e., $\left|\hat{T}_{\hat{t}}^{(n,\omega)}\right| > \Phi(1-\alpha/2)$), and the alarm is not a false alarm, then the last observation read by our algorithm will be at $\tilde{t} = \hat{t} + \omega$. In this case, if the true change point is located at $t^*$ and the refined location of the estimated change point is at $\hat{\hat{t}}$, then the detection delay and refined localization error will be $\tilde{t} - t^*$ and $|\hat{\hat{t}} - t^*|$, respectively. Similar to the previous simulations, we consider various values for the dimension of data and the jump size of the change point. Additionally, we introduce the refine size, which takes values of $1/2$, $1/5$, $1/10$, and $1/50$. The data points are generated with a total length of $2600$, and the change point is located at position $2300$. We estimate the transition matrices and variances using the first $2000$ observations. Subsequently, we apply our algorithm to the remaining data points with $\alpha$ set to $1/1000$ and $\omega$ set to $50$, both with and without the confirmation step introduced at the end of Section~\ref{refine}. This entire process is repeated $200$ times, during which we calculate the early stop rate and record the detection delays and refined localization errors for all combinations.

Figure~\ref{fig:rs01} presents the summarized box plots for detection delays and refined localization errors when the refine size is set to $0.1$ for all combinations of data dimensions and jump sizes without the confirmation step. As depicted in the figure, the refinement process effectively reduces the localization error, specially when the jump size is relatively large. As illustrated in Figure~\ref{fig:rsearlystop}, the confirmation step notably decreases the possibility of false alarms. Thus, the confirmation step can be considered as an option to minimize false alarm probabilities. To provide practical guidance on the choice of refine size based on the window size recommendation \( \omega = 10\log(hp^2) \) in Section~\ref{sim:w}, we conducted a sensitivity analysis. Specifically, in each experimental iteration, we simulated scenarios in which alarms were triggered using \( \omega = 10\log(hp^2) \) observations, with the true change point positioned at the center of this larger window. The refinement was then applied using refine sizes (\(0.1, 0.2, 0.3, 0.4, 0.5\)). This analysis was conducted across various data dimensions (\( p = 10, 40, 70, 100 \)) and jump sizes (2, 3, 4), with the goal of identifying the refine size that minimized refined localization error. After 100 repetitions, the average optimal refine sizes consistently clustered around 0.1 to 0.2, as shown in Table~\ref{tbl:refine}. Based on these results, we recommend using \( 0.15 \) for practical applications.

\begin{table}[!ht]
    \centering
    \begin{tabular}{|c|c|c|c|c|}
    \hline
    Jump Size & \( p = 10 \) & \( p = 40 \) & \( p = 70 \) & \( p = 100 \) \\ 
    \hline
    2 & 0.131 & 0.138 & 0.149 & 0.157 \\ 
    3 & 0.116 & 0.126 & 0.137 & 0.157 \\ 
    4 & 0.113 & 0.107 & 0.119 & 0.116 \\ 
    \hline
    \end{tabular}
    \caption{Average of the optimal refine sizes for different dimensions and jump sizes.}
    \label{tbl:refine}
\end{table}

\begin{figure}[!ht]
\centering
\includegraphics[width=0.8\textwidth]{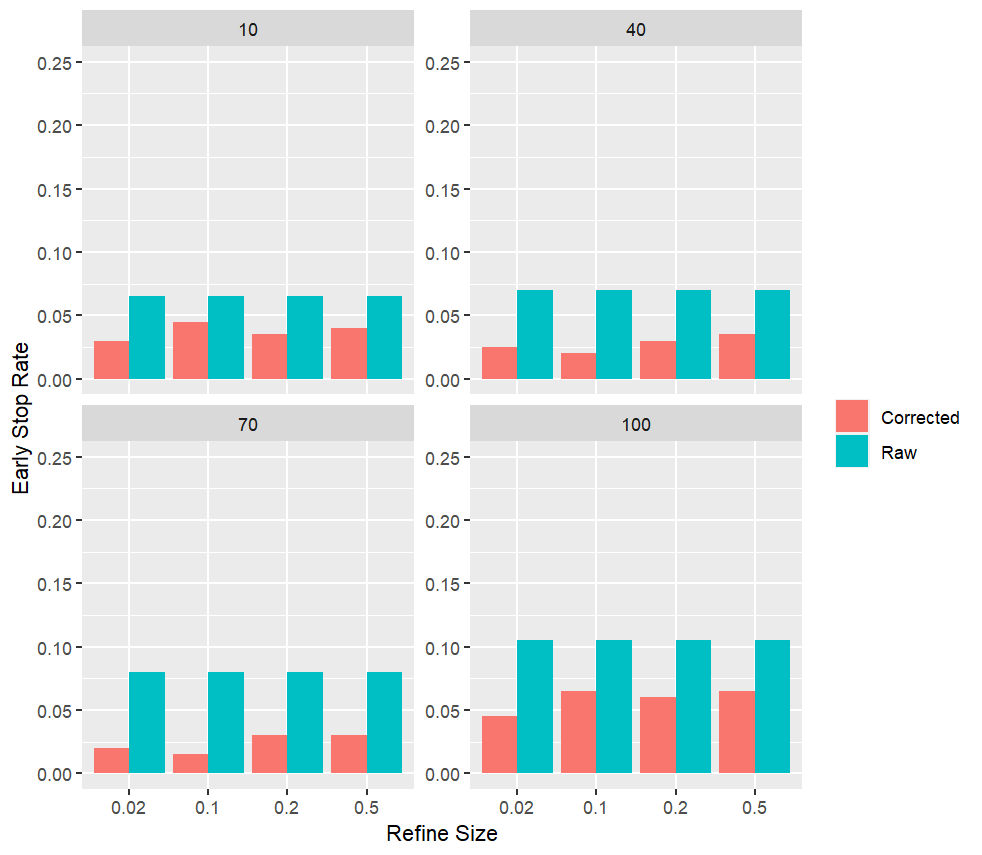}
\caption{Simulation D: This plot provides a summary of the early stop rates for all combinations of refine sizes, data dimensions and whether confirmation is used or not. For each grid in the plots, the dimension of data is set to $10$, $40$, $70$, and $100$.}\label{fig:rsearlystop}
\end{figure}

\begin{figure}[!ht]
\centering
\includegraphics[width=0.8\textwidth]{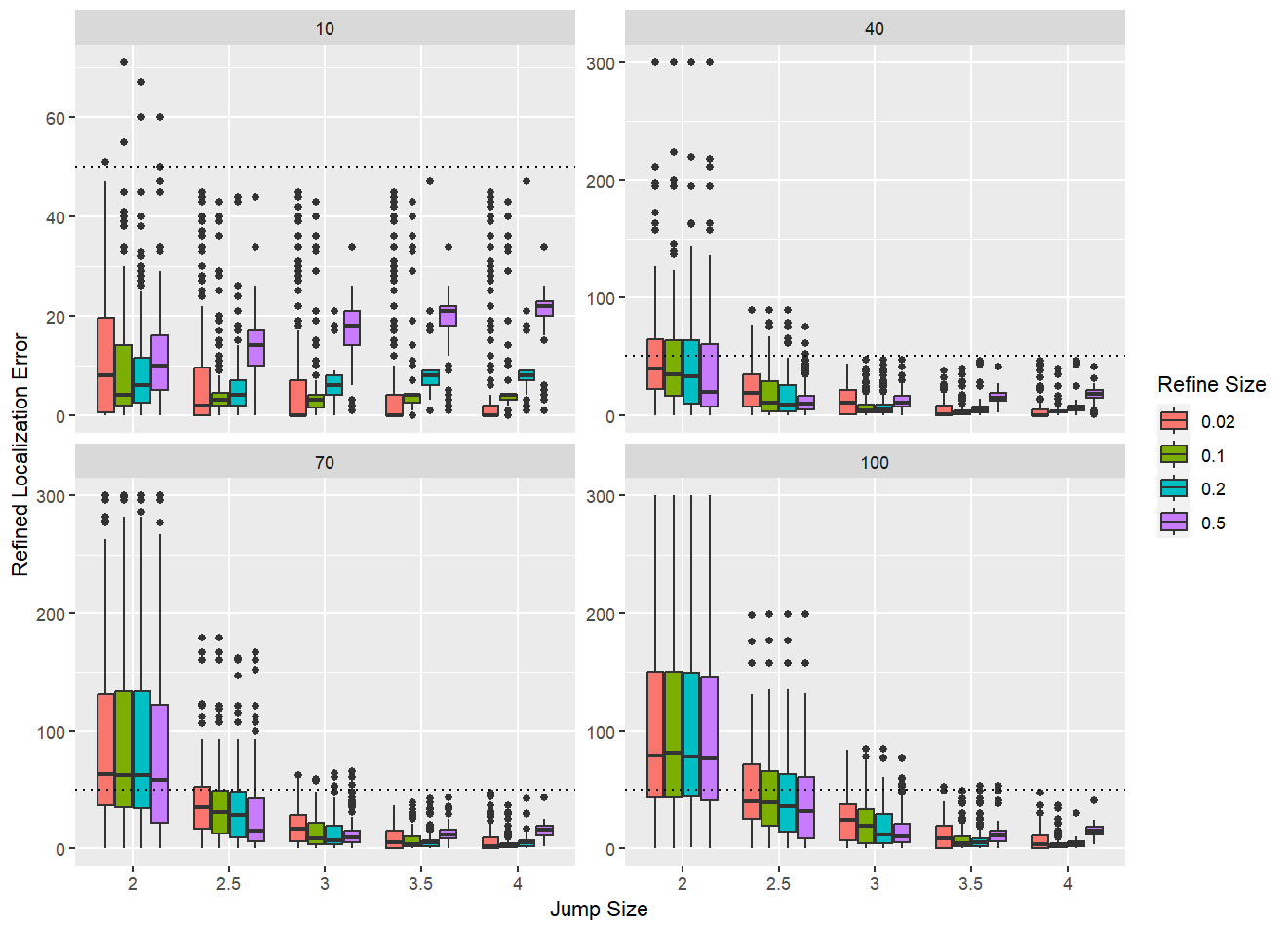}
\caption{Simulation D: This plot provides a summary of the refined localization errors for all combinations of refine sizes and data dimensions. For each grid in the plots, the dimension of data is set to $10$, $40$, $70$, and $100$. The horizontal dashed line represents the pre-specified detection delay, denoted as $\omega$.}\label{fig:rsrefinedle}
\end{figure}

\subsection{Simulation E: Multiple Change Point Detection}\label{sim:mcp} In this simulation, we evaluate the performance of our method in terms of F1 score when dealing with multiple change points in low-dimensional ($p=10$) and high-dimensional ($p=100$) setups. For a range of jump sizes, we generate VAR time series of size $6900$ with change points located at positions $2300$ and $4600$. Specifically, for the first $2300$ data points, we use the transition matrix $0.8*I_p$. The subsequent $2300$ data points are generated using a new transition matrix with a certain jump size compared to the previous one. Finally, the last $2300$ data points are generated again using the transition matrix $0.8*I_p$. We implement Algorithm~\ref{alg1} sequentially, as mentioned in Section~\ref{multicp}, and consider the refined estimated change points within $2300 \pm 10$ and $4600 \pm 10$ as true positives. In each repetition, we calculate the following metrics:
$F1\,Score = \frac{2 \times TP}{2 \times TP + FP + FN}$ where TP, FP, and FN represent true positives, false positives, and false negatives, respectively. We then calculate the averages among the 100 repetitions for different jump sizes. These metrics are commonly used in assessing detection algorithms in scenarios with multiple change points such as in \cite{bai2023unified}. The results are summarized in Table~\ref{tbl:multi}. Under both low-dimensional and high-dimensional setups, we set $n = 2000$, $\omega = 50$, $\alpha = 0.0001$, and $h = 1$ for our algorithm. To reduce the number of false alarms, we perform the confirmation step as introduced in Section~\ref{refine}. As shown in Table~\ref{tbl:multi}, our algorithm exhibits strong capabilities in handling data with multiple change points, especially when the jump size is large, under both low-dimensional and high-dimensional setups.

\begin{table}[ht]
\caption{Simulation E: The F1 score for our algorithm is assessed in a multiple change point scenario. We consider jump sizes (JS) ranging from $2$ to $4.5$ under both low-dimensional ($p=10$) and high-dimensional ($p=100$) setups.}
\label{tbl:multi}
\begin{center}
\begin{tabular}{|c|c|c|c|c|c|c|}
\hline
& JS = 2.0 & JS = 2.5 & JS = 3.0 & JS = 3.5 & JS = 4.0 & JS = 4.5 \\
\hline
$p=10$ & 0.73 & 0.88 & 0.97 & 0.98 & 0.99 & 0.99 \\
\hline
$p=100$ & 0.06 & 0.26 & 0.45 & 0.66 & 0.88 & 1.00 \\
\hline
\end{tabular}
\end{center}
\end{table}

\subsection{Simulation with Variance Heterogeneity}\label{sim:hetervar}
This section provides simulation results for the average run length and detection delay of our algorithm under the same setup as in Simulation A and B, with $\alpha = 1/1000$. However, in this simulation, the diagonal entries of the covariance matrix for the noise is randomly generated from a uniform distribution ranging from $0.5$ to $1.5$ to assess our algorithm's performance with variance heterogeneity. The test statistic is calculated as described in Remark~1. Satisfactory performance is achieved for both average run length and detection delay in this scenario, as shown in Figure\ref{fig:heter}.
\begin{figure}[!ht]
\centering
\includegraphics[width=1\textwidth]{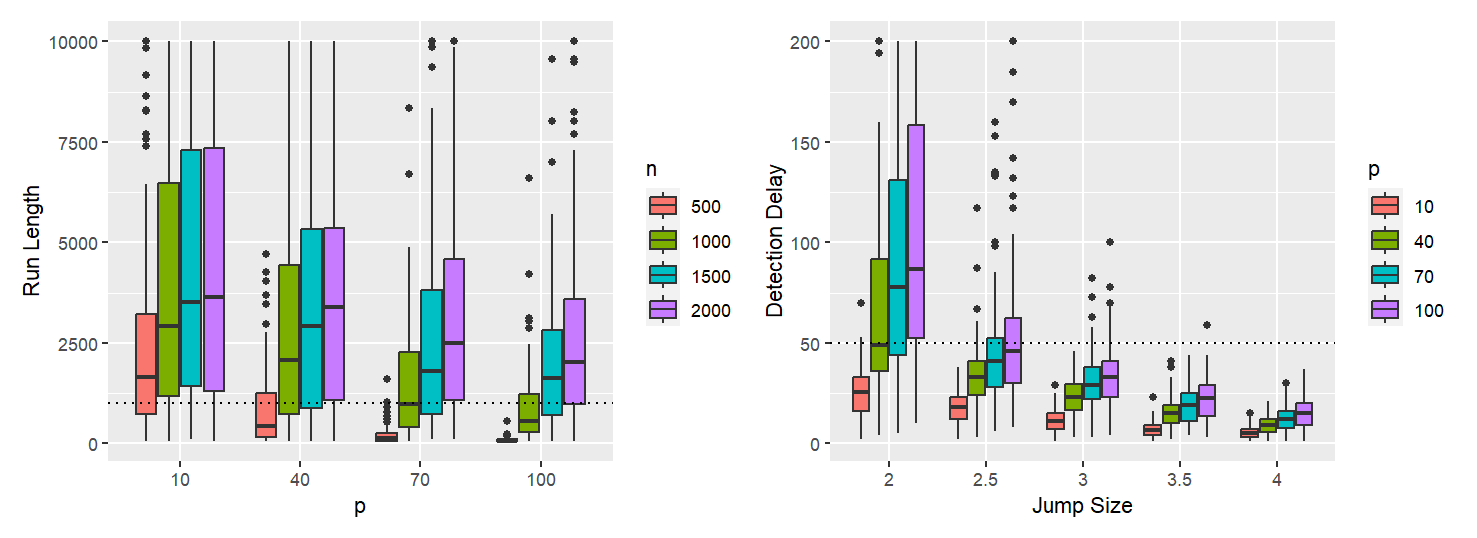}
\caption{Simulation with variance heterogeneity: The covariance matrix's diagonal elements for errors are randomly selected from a uniform distribution ranging from 0.5 to 1.5. The rest of the settings align with those of Simulation A and B, with the value of $\alpha$ set to $1/1000$.}\label{fig:heter}
\end{figure}

\subsection{Numerical Comparison in High-Dimensional Settings}\label{sim:hdcompare}

This section supplements Section~\ref{section:num_main} by extending the numerical comparison to high-dimensional settings with \( p = 100 \). In addition to this modification, we increased the training sample size from 500 to 2000 and adjusted the jump sizes from 2 and 3 to 3 and 4 to accommodate the higher dimensionality. The results, shown in Figure~\ref{fig:compare100}, demonstrate that our proposed algorithm performs comparably to the case when \( p = 10 \). Notably, the algorithm remains competitive with alternative methods when the data is generated without a VAR structure and continues to outperform all competing methods when the data is generated with a VAR structure. Additionally, we observed that the TSL method \citep{qiu2022transparent} required an excessive amount of memory (over 8,388,608 GB) to allocate the necessary vectors in the larger dimensional setting. As a result, we were unable to obtain results for the TSL method in this scenario, and it is therefore not included in the comparison.

\begin{figure}[!ht]
    \centering
    \includegraphics[width=1\linewidth]{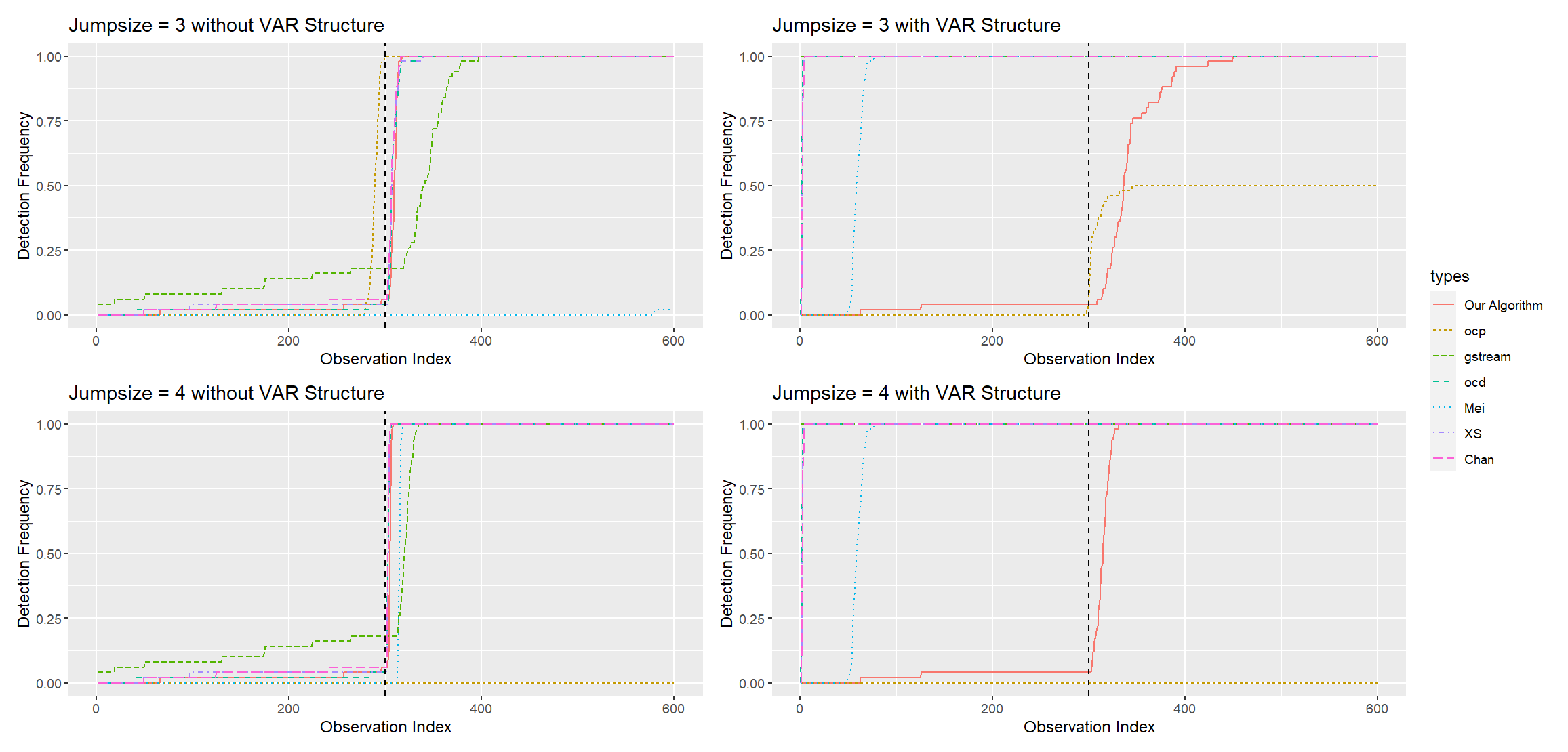}
    \caption{Summary of detection frequencies for all algorithms. The black dashed vertical line indicates the location of the true change point. An ideal algorithm would demonstrate a detection frequency of zero before the line and achieve one immediately after the line.}
    \label{fig:compare100}
\end{figure}

\subsection{Robustness to Time-Varying Transition Matrices}\label{sim:timevarying}

To illustrate the robustness of the proposed algorithm to small time-varying effects, we conducted a set of simulations with a transition matrix that varies slightly over time. These simulations, summarized in Figure~\ref{fig:varying}, involved introducing time-varying behavior in three specific entries of the transition matrix. In the left panel of the figure, the entry in row 2, column 2 oscillates between \( 0.5 \pm 0.3 \) with a period of 500. The other two time-varying entries oscillate similarly but start from different initial values. These oscillations persist throughout the simulation, even after change points. We varied the amplitude of oscillation across different runs, testing values of 0, 0.1, 0.2, and 0.3, where 0 represents no time-varying effect. Two sets of simulations were conducted to examine the algorithm's performance under these conditions. The first set of simulations, shown in the middle panel, evaluated how changes in amplitude affect the run length. With settings similar to those in \ref{sim:runlength}—using \( \alpha = 1/1000 \), \( n = 500 \), \( p = 10 \), and \( \omega = 50 \)—the results show that, for small oscillation amplitudes, the algorithm maintains control over the target ARL, keeping it above \( 1/\alpha \). However, as the amplitude increases, the run length decreases, indicating that larger time-varying effects are more likely to be misidentified as true changes, leading to a higher false alarm rate. The second set of simulations, shown in the right panel, analyzed the effect of oscillation amplitude on detection delay. Under settings similar to those in \ref{sim:detectiondelay} (with \( \alpha = 1/1000 \), \( n = 500 \), \( p = 10 \), and a jump size of 2), the results indicate that the detection delay remains relatively stable, even as oscillation amplitude increases. This demonstrates that the detection delay is less sensitive to moderate time-varying effects. In summary, while the full extension of this method to handle time-varying transition matrices lies beyond the scope of this study, these simulations show that the proposed algorithm is robust to small time-varying effects. Future research will further explore this aspect. For now, the focus of this work remains on the piecewise constant setting.

\begin{figure}[!ht]
    \centering
    \includegraphics[width=1\textwidth]{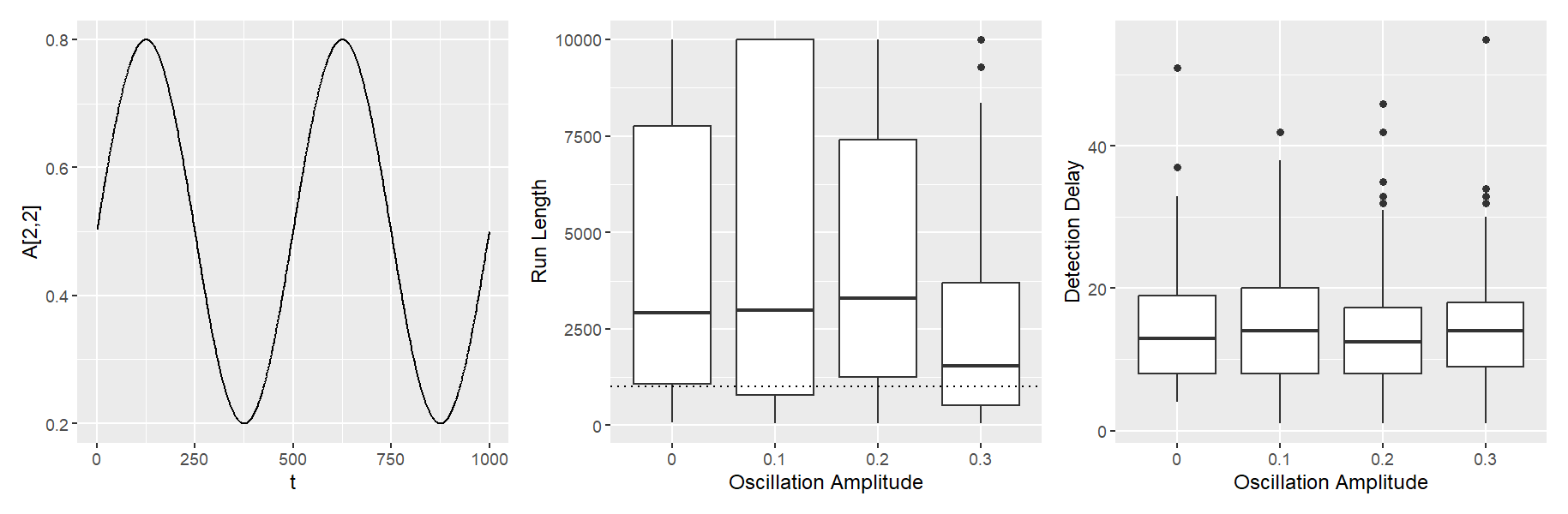}
    \caption{Simulation with Time-Varying Transition Matrices. (Left) Illustration of how specific entries in the transition matrix vary over time. The entry at row 2, column 2 oscillates between 0.5 $\pm$ 0.3 with a period of 500. (Middle) Effect of oscillation amplitude on the run length. The dashed line indicates the target ARL of \(1/\alpha\), with run lengths expected to exceed this threshold. (Right) Effect of oscillation amplitude on the detection delay.}\label{fig:varying}
\end{figure}

\subsection{Robustness to Complex Transition Matrix Structures}\label{sim:nonsparse}

The proposed algorithm is capable of handling more complex transition matrices, provided there is a sufficiently large training sample size to enable accurate estimation. To illustrate its robustness, we conducted additional simulations using a low-rank plus sparse structure for the transition matrix, following the setup described in \cite{bai2020multiple}. In this simulation, the transition matrix includes a low-rank component with a rank of 2, resulting in a structure that is no longer sparse. The simulation parameters were set to \( p = 25 \), \( \alpha = 1/1000 \), and \( \omega = 50 \), and Figure~\ref{fig:lps} summarizes the results. As shown in Figure~\ref{fig:lps}, the false alarm rate remains well-controlled when the sample size is sufficiently large. However, more observations are required to ensure that the average run length (ARL) meets the target threshold of \( 1/\alpha = 1000 \) when handling complex transition matrices. Notably, the detection delay appears to be more sensitive to the magnitude of the jump than to the structure of the transition matrix itself. The settings for these simulations align with those used in Sections \ref{sim:runlength} and \ref{sim:detectiondelay}, where we analyze the run length and detection delay under different scenarios. While these results demonstrate the algorithm’s capability to handle more intricate transition matrix structures, we do not pursue a rigorous theoretical analysis of this aspect here. Instead, we aim to provide an empirical illustration of the algorithm's robustness, leaving a deeper theoretical investigation for future work.

\begin{figure}[!ht]
    \centering
    \includegraphics[width=1\linewidth]{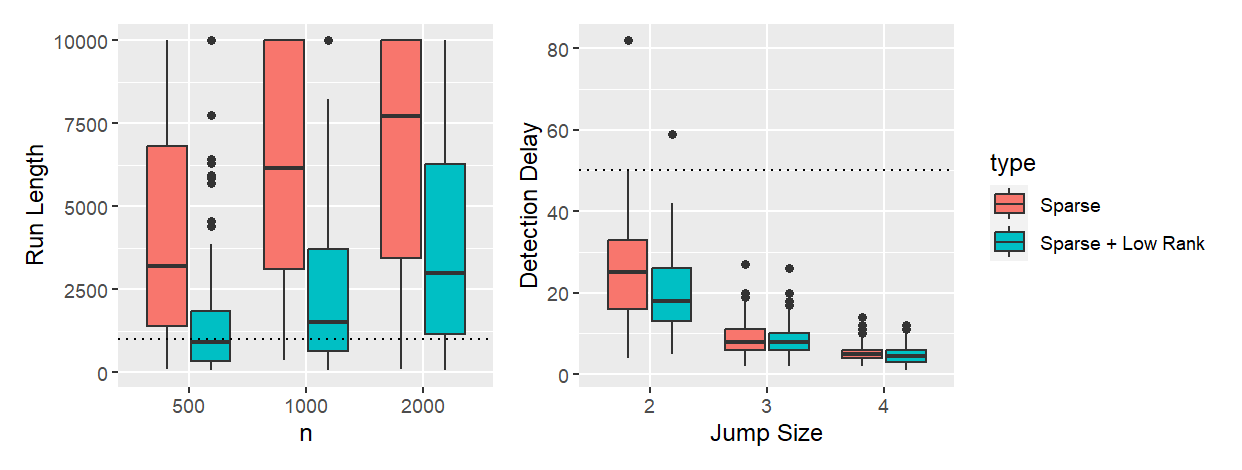}
    \caption{Performance of the Proposed Algorithm in Terms of Run Length and Detection Delay with Sparse vs. Sparse + Low-Rank Transition Matrices.}
    \label{fig:lps}
\end{figure}

\section{ADDITIONAL RESULTS FOR REAL DATA ANALYSIS}\label{addreal}
This section provides additional details for the S\&P 500 real data experiment and presents results from the real data experiment conducted on EEG data.

\subsection{Additional Details for S\&P 500 Data}\label{addsp500}

To establish a reference for the anomaly period, the return volatility is used, a standard measure of return dispersion (also used as a reference in \cite{keshavarz2020sequential}). Let \( x_{t,j} \) represent the daily log return for stock \( j \) at time \( t \), and let \(\operatorname{std}(x)\) denote the standard deviation of \(x\). The return volatility of stock \( j \) at time \( t \) is estimated using the formula \( z_{t,j} = \operatorname{std}(x_{t,j}, \ldots, x_{t+\omega-1,j}) \). The average \( z_{t,j} \) across all 186 stocks is then computed, and this average return volatility is rescaled for visualization. The rescaled value is shown as the black line in Figure~\ref{fig:real}. A high average return volatility generally indicates an increased likelihood of a change point. Figure~\ref{fig:real} also shows the locations of alarms (red vertical lines) and the estimated onsets (black vertical lines) of alarm clusters.

\begin{figure}[!ht]
\centering
\includegraphics[width=0.8\textwidth]{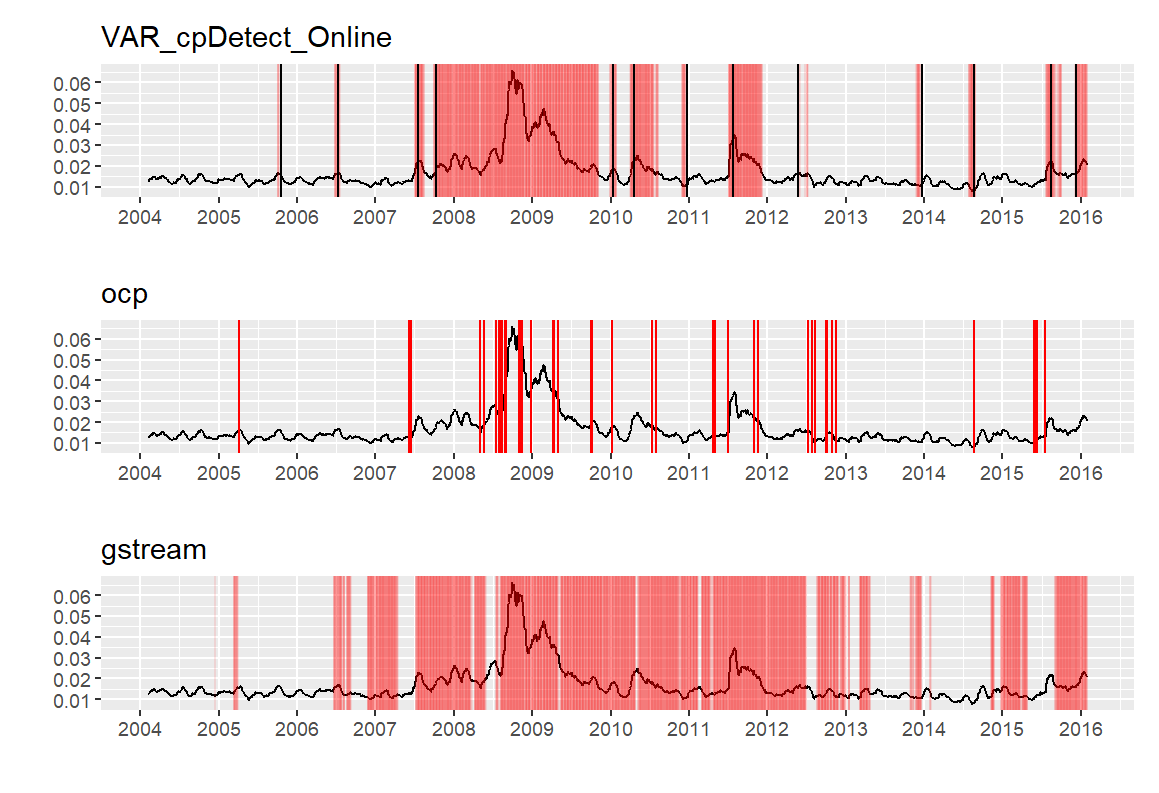}
\caption{Experiment results on S\&P 500 data: The black line represents the rescaled average return volatility, while the red lines correspond to alarm locations for (top) VAR\_cpDetect\_Online, (middle) ocp, and (bottom) gstream.}\label{fig:real}
\end{figure}

\subsection{Real Data Experiment on EEG Data}\label{eeg}
For the EEG data, this experiment aims to detect and raise an alarm indicating an impending seizure, occurring around \( t = 85 \), as confirmed by neurologists and validated by offline change point detection methods in Section~\ref{realdata} of \cite{safikhani2022joint}. The data was collected from 18 EEG channels over a 227.68-second duration. To focus on seizure onset, data after \( t = 150 \) seconds were removed. The final dataset comprises 1500 data points over 150 seconds, with a dimension of \( p = 18 \). The first 300 data points were designated as historical data, with parameters \( \omega = 30 \) and \( \alpha = 1/2000 \) used in our method. The hyperparameters for the baseline algorithms were selected as described in Section~\ref{section:num_main}. All methods were applied to the entire dataset without halting upon alarm, and the alarm locations are documented in Figure~\ref{fig:EEG}.
\begin{figure}[!ht]
\centering
\includegraphics[width=1\textwidth]{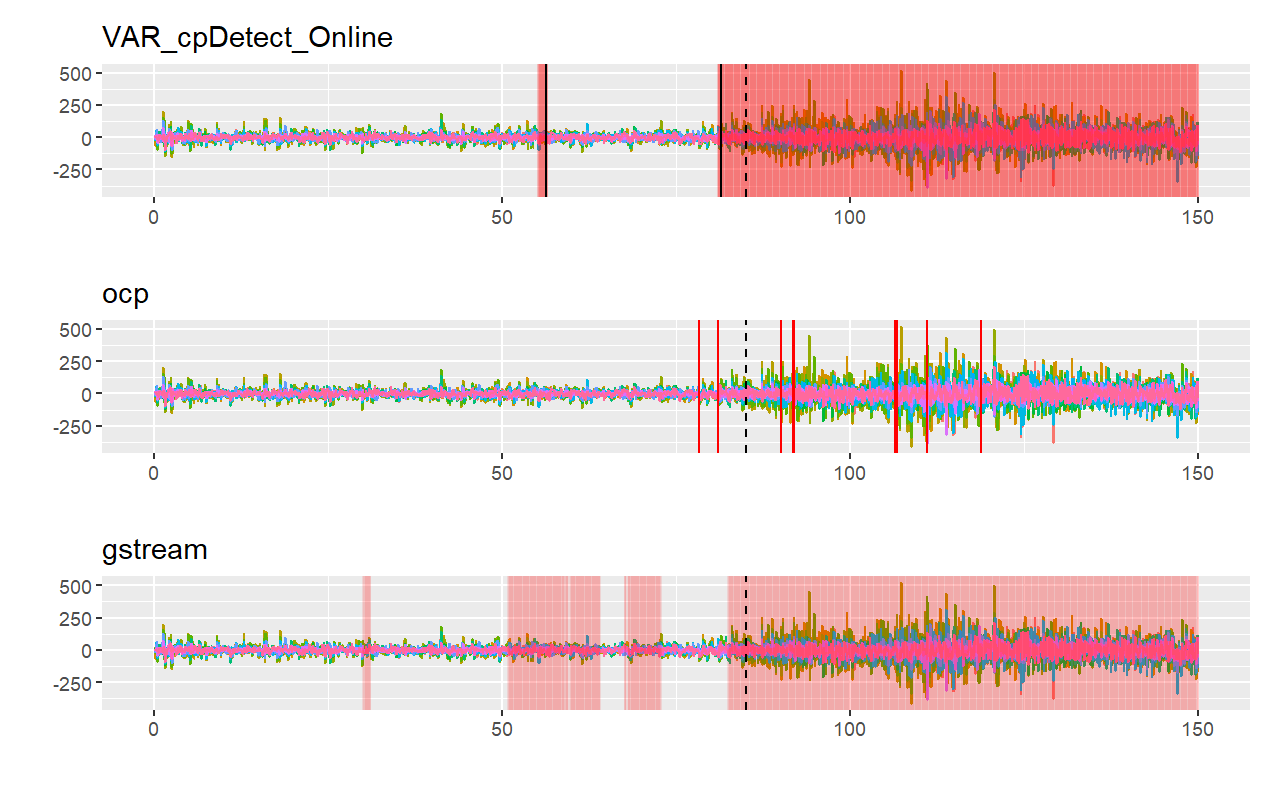}
\caption{Experiment results on EEG data: The red lines indicate alarm locations for (top) VAR\_cpDetect\_Online, (middle) ocp, and (bottom) gstream.}\label{fig:EEG}
\end{figure}
As shown in the top panel of Figure~\ref{fig:EEG}, the alarms raised by the proposed algorithm form two clusters, indicating periods where the patient's brain activity deviates from baseline, potentially signaling seizure activity. The estimated start times (solid vertical black lines) of these clusters are at \( t = 56.3 \) (lasting 1.4 seconds) and \( t = 81.5 \) (lasting until the end of the data). The proposed algorithm requires 19 and 25 additional observations (detection delay) before issuing these alarms. Both estimates occur before the confirmed seizure onset at \( t = 85 \), suggesting that early shifts in brain electrical activity may be detectable in advance, consistent with findings in \cite{ombao2005slex}. Similarly, in another study \citep{safikhani2022joint}, an offline CPD algorithm based on a VAR model estimated a change point at $t = 83$, also slightly before the seizure began, further supporting the idea that changes in brain activity may be detectable prior to the seizure's onset. The middle panel shows that the ocp method detected two change points at \( t = 78.3 \) and \( t = 81.0 \), both preceding the seizure onset, with the latter closely aligning with our algorithm's estimate. The bottom panel indicates that the gstream method raised alarms forming four clusters, with start points at \( t = 30.1 \) (lasting 1.1 seconds), \( t = 50.9 \) (lasting 13.3 seconds), \( t = 67.6 \) (lasting 5.3 seconds), and \( t = 82.5 \) (lasting until the end). The final cluster occurs slightly before \( t = 85 \), agreeing with our algorithm's results; however, the gstream method triggers numerous alarms before the seizure, limiting its practical utility for early warning. The average execution times were 1.77 seconds for our method, 9.50 seconds for ocp, and 27.14 seconds for gstream.

\section{SEQUENTIAL UPDATING FOR TRANSITION MATRICES}\label{app:sequpdate}
In this section, we examine the performance of a sequential updating approach \citep{messner2019online} for estimating transition matrices in high-dimensional VAR models. Sequential updating allows for efficient integration of new data, improving the estimation of transition matrices for the proposed algorithm when no alarm has been raised during monitoring. We discuss the benefits and limitations of sequential updating in various scenarios and present simulation results to illustrate its impact on estimation accuracy.

\subsection{Advantages and Limitations}

Sequential updating provides a practical method to update the transition matrix estimates as new observations arrive. Instead of re-estimating the transition matrices from scratch using both old and new data, which incurs high time and space complexity, this approach applies a cyclic coordinate descent algorithm at each time step. By using the previous step’s coefficient estimates as starting values, it avoids the computational burden associated with full re-estimation. When the forgetting factor is set to \( \nu = 1 \), this approach efficiently updates the transition matrices without requiring all historical data. The detailed procedure can be found in \cite{messner2019online}, particularly in Equations (10)–(14).

Incorporating sequential updating into the proposed algorithm is particularly advantageous when the training data is very limited. In such cases, it allows the transition matrix estimates to improve as additional observations are gathered, provided no alarm is raised. This can help the algorithm reduce false alarm rates and increase power, as the accuracy of the transition matrix estimates improves with more data. However, as the size of the initial training data grows, the relative benefits of sequential updating decrease.

One limitation of sequential updating is that it can introduce estimation error at the beginning of the process, which may increase the likelihood of false alarms. This issue is particularly critical in real-time applications where accuracy is essential. As a result, while sequential updating is valuable in cases with limited training data, its direct application may not be suitable for all scenarios, especially when minimizing false alarm rates is crucial.

\subsection{Simulation Study}

To illustrate the effects of sequential updating, we conducted simulations using data generated from a VAR process with transition matrix \( A \) and dimension \( p = 10 \). For each repetition, the initial estimate of the transition matrix was obtained using the regularization method \citep{basu2015regularized} on training data of varying lengths, followed by sequential updates as new observations were collected.

Figure~\ref{fig:training_data_comparison} shows the estimation error \( \|A - \hat{A}\|_2 \) with and without sequential updating. In this figure, the solid black line represents the average estimation error using sequential updating, with black dotted lines indicating the 2.5th and 97.5th percentiles. The blue dotted line represents the average estimation error when only the initial training data is used, with the shaded area showing the 2.5th and 97.5th percentile range.

The results show that sequential updating initially increases the estimation error, which may temporarily raise the false alarm probability. As the training sample size grows, however, the advantage of sequential updating diminishes, and the overall estimation error converges with that of the non-updated estimates. 

The simulation results suggest that sequential updating is advantageous in situations with limited training data, providing an efficient way to incorporate new data and improve estimation accuracy. However, as the amount of data increases, the benefits of sequential updating wane, and its initial estimation error may contribute to a higher false alarm risk. Therefore, while sequential updating is effective in specific scenarios, we advise caution in applying it directly in cases where minimizing false alarms is a priority. Integrating sequential updating with the proposed algorithm presents an interesting but challenging direction for future research.

\begin{figure}[!ht]
\centering
   \begin{subfigure}[b]{0.4\textwidth}
   \centering
   \includegraphics[width=1\textwidth]{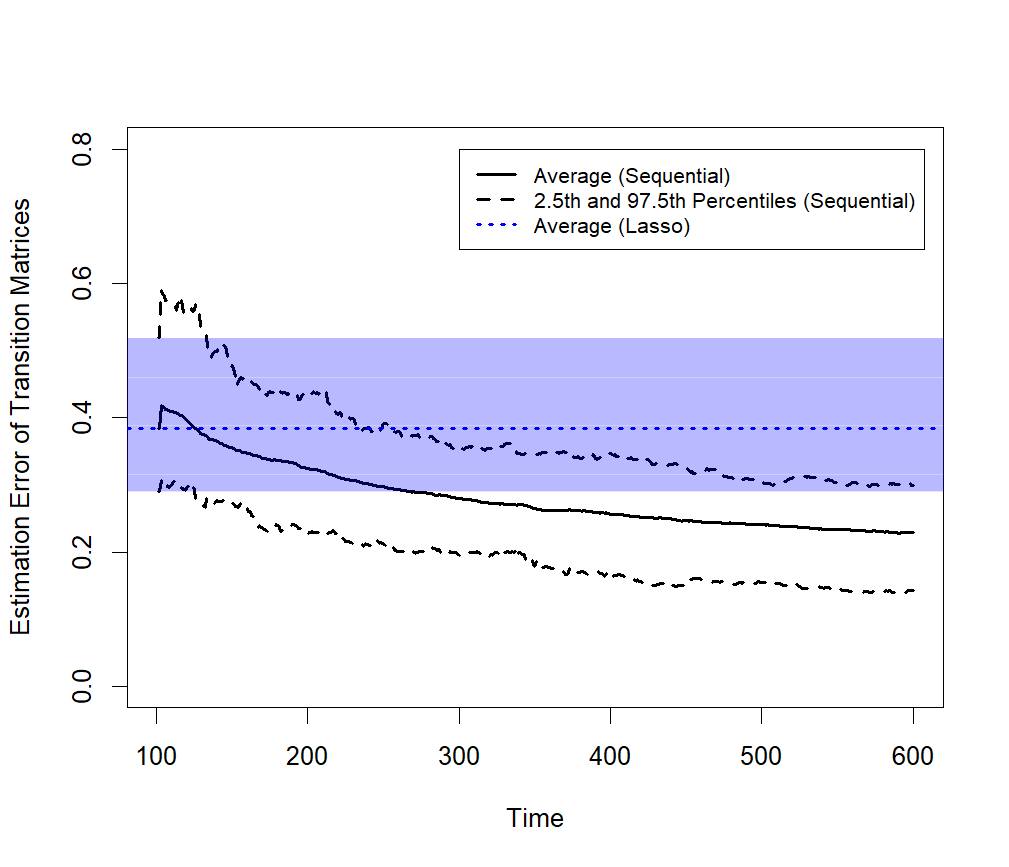}
   \caption{Training data length = 100}
   \label{fig:100} 
\end{subfigure}%
\begin{subfigure}[b]{0.4\textwidth}
\centering
   \includegraphics[width=1\textwidth]{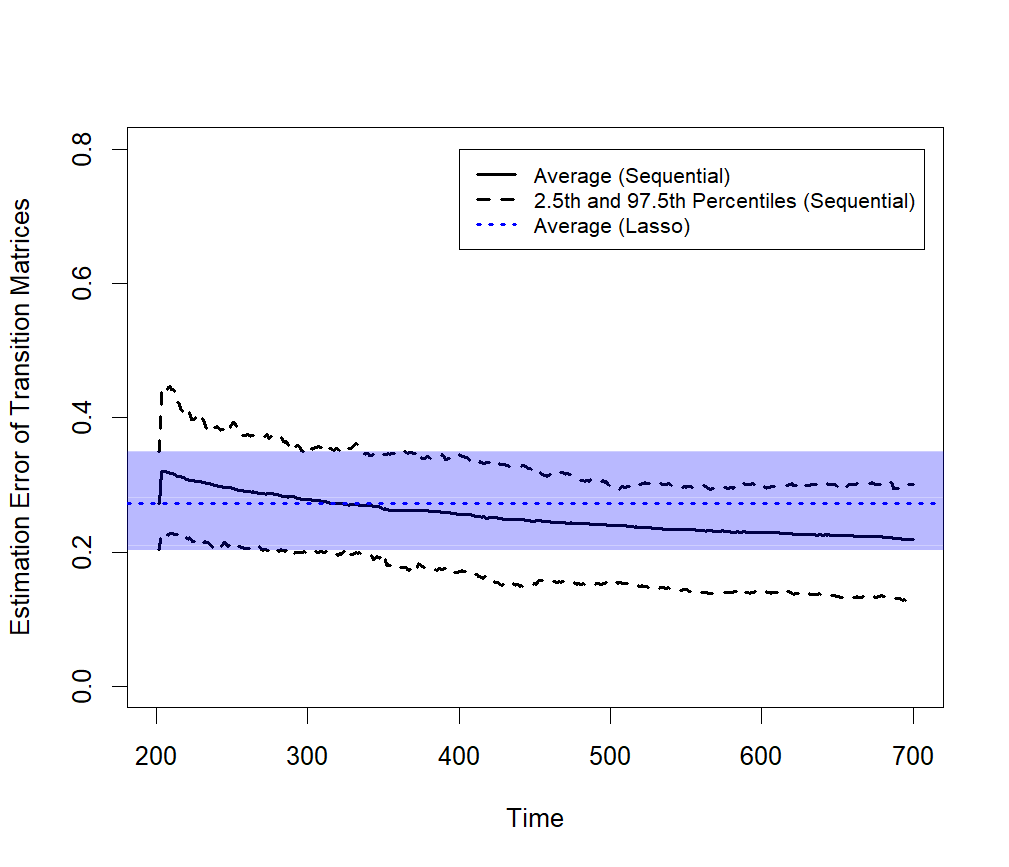}
   \caption{Training data length = 200}
   \label{fig:200}
\end{subfigure}
\begin{subfigure}[b]{0.4\textwidth}
\centering
   \includegraphics[width=1\textwidth]{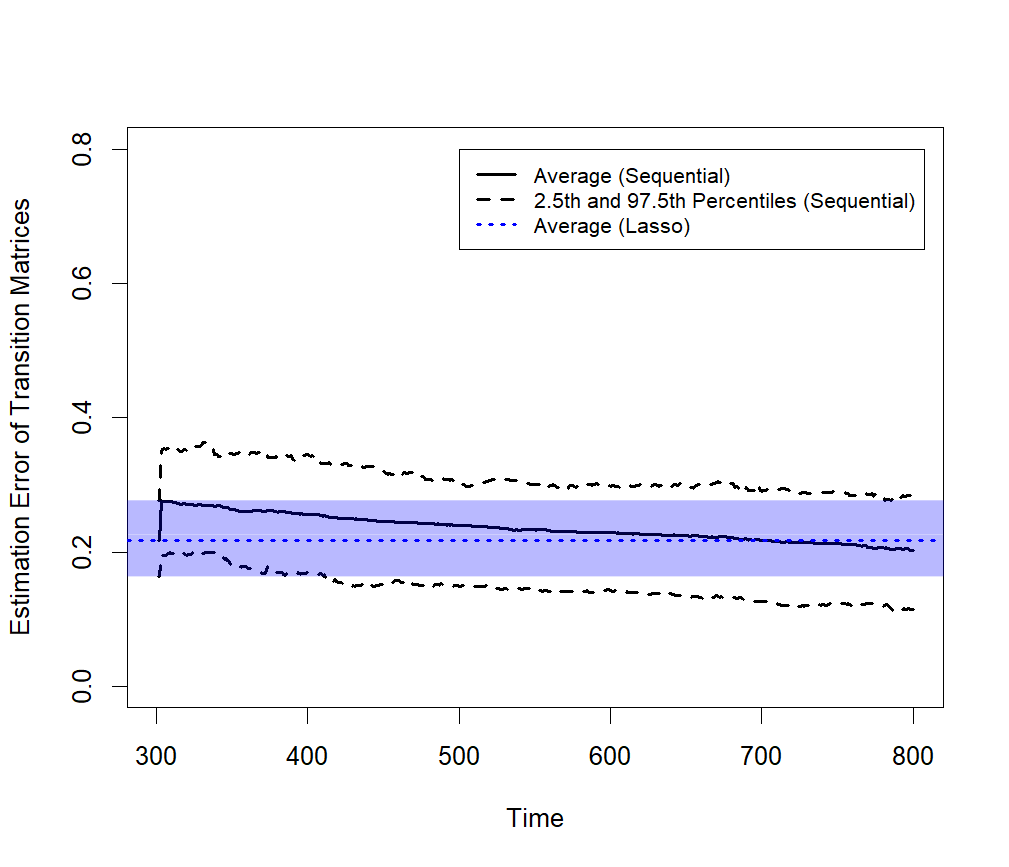}
   \caption{Training data length = 300}
   \label{fig:300}
\end{subfigure}%
\centering
   \begin{subfigure}[b]{0.4\textwidth}
   \centering
   \includegraphics[width=1\textwidth]{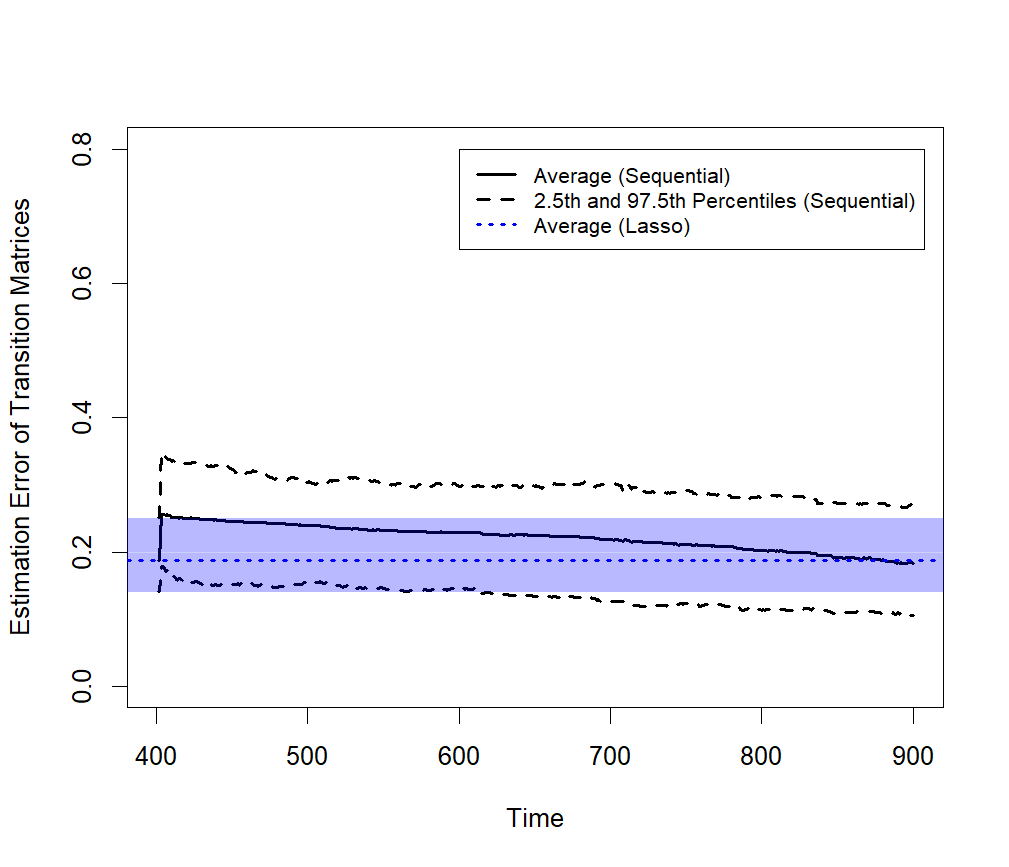}
   \caption{Training data length = 400}
   \label{fig:400} 
\end{subfigure}

\caption{Comparison of estimation errors between the sequential update method and estimates based only on training data}
\label{fig:training_data_comparison}
\end{figure}

\section{POST-CHANGE ANALYSIS}\label{postchangeanalysis}
Identifying which variables undergo shifts after a change point, especially in high-dimensional contexts, is important yet challenging due to the limited number of post-change samples. When the post-change sample size is small, estimating the new model parameters reliably becomes difficult, complicating efforts to pinpoint which variables have shifted. Even with an accurately detected change point, a limited number of post-change observations can greatly reduce the reliability of diagnostic analysis. A straightforward approach might be to estimate the transition matrices before and after the change using Lasso, then compare these estimates. However, the bias inherent in Lasso makes it infeasible to directly infer which components of the transition matrices have changed. To address this, we recommend applying an online debiasing technique  \citep{deshpande2023online} both before and after the change point. This method debiases the Lasso estimates and allows for constructing confidence intervals (CIs) for the entries of the VAR transition matrices. By constructing CIs for the differences between the debiased estimates of the transition matrices before and after the change, we can identify which entries are likely to have changed. If the CI for a given entry excludes zero, we can infer a significant shift in that entry. To validate this approach, we conducted simulations with two groups of observations with \( p = 10 \) variables—one representing data before the change (with transition matrix \( A \)) and the other representing data after the change (with transition matrix \( A^* \)). The pre-change sample size was fixed at \( n_0 = 500 \), allowing for accurate estimation, while the post-change sample size \( n_1 \) varied among 100, 200, and 300. The two transition matrices differed at six specific entries: (1,1), (2,2), (10,10), (3,7), (6,4), and (8,4). For each entry in the difference matrix \( D = A - A^* \), we constructed CIs using debiased Lasso estimates and calculated their coverage rates of zero over 100 repetitions. As shown in Figure~\ref{fig:diagnosis}, entries with no changes maintain a zero coverage rate around 0.95, while entries with changes rarely include zero as the post-change sample size increases, accurately identifying the shifts. Although identifying these changes benefits from a moderate number of post-change samples, the CPD algorithm can continue collecting observations to improve diagnostic accuracy.

\begin{figure}[!ht]
    \centering
    \includegraphics[width=1\linewidth]{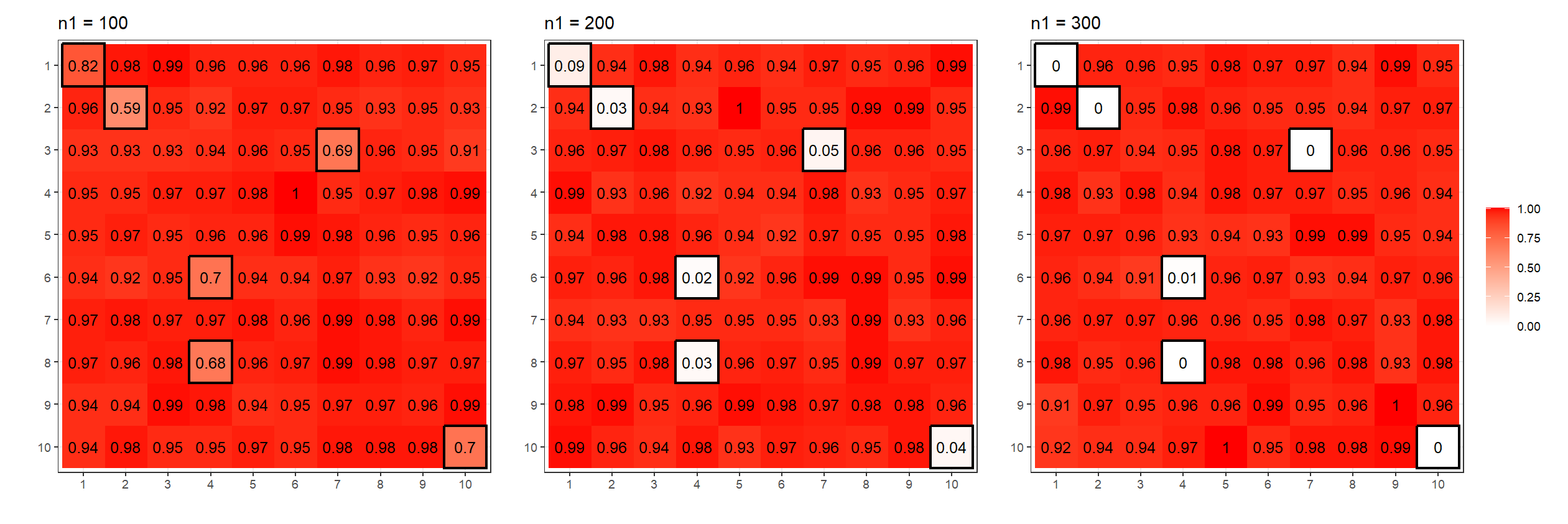}
    \caption{Coverage rates of zero for confidence intervals of the entries in $A - A^*$}
    \label{fig:diagnosis}
\end{figure}

\end{document}